\documentclass[10pt,draftclsnofoot,onecolumn]{IEEEtran}

\usepackage{cite}

\usepackage{graphicx}
\usepackage{subfigure}
\usepackage{epstopdf}
\usepackage{multicol}
\usepackage{algorithm}
\usepackage{algorithmic}
\usepackage{amsthm,amsmath,amsfonts}
\usepackage{mathrsfs}
\usepackage{courier}
\newtheorem{lemma}{\textbf{Lemma}}

\newtheorem{proposition}{\textbf{Proposition}}

\newtheorem{mydef}{\textbf{Definition}}
\newtheorem{assumption}{\textbf{Assumption}}
\usepackage{bm}
\usepackage{subfigure}

\usepackage{color}

\begin{document}
%
\title{Proactive UAV Network Slicing for URLLC and Mobile Broadband Service Multiplexing}
%
%
%
%


\author{Peng Yang, Xing Xi, Kun Guo, Tony Q. S. Quek,~\IEEEmembership{Fellow,~IEEE}, Jingxuan Chen, and

Xianbin Cao,~\IEEEmembership{Senior Member,~IEEE}
\thanks{
P. Yang, K. Guo and T. Q. S. Quek are with the Information Systems Technology and Design, Singapore University of Technology and Design, 487372 Singapore.

X. Xi, J. Chen, and X. Cao are with the School of Electronic and Information Engineering, Beihang University, Beijing 100083, China, and also with the Key Laboratory of Advanced Technology, Near Space Information System (Beihang University), Ministry of Industry and Information Technology of China, Beijing 100083, China.

This paper was presented in part in the IEEE Global Communications Conference 2020 \cite{yang2020repeatedly}.
}
}

\maketitle

\begin{abstract}
The unmanned aerial vehicle (UAV) network that is convinced as a significant component of 5G and emerging 6G wireless networks is desired to accommodate multiple types of service requirements simultaneously. However, how to converge different types of services onto a common UAV network without deploying an individual network solution for each type of service is challenging. We tackle this challenge in this paper through slicing the UAV network, i.e., creating logical UAV networks customized for specific requirements. To this end, we formulate the UAV network slicing problem as a sequential decision problem to provide mobile broadband (MBB) services for ground mobile users while satisfying ultra-reliable and low-latency requirements of UAV control and non-payload signal delivery. This problem, however, is difficult to be directly solved mainly due to the sequence-dependent characteristic and the lack of accurate location information of mobile users and accurate and tractable channel gain models in practice. To overcome these difficulties, we propose a novel solution approach based on learning and optimization methods. Particularly, we develop a distributed learning method to predict mobile users' locations, where partial user location information stored on each UAV is utilized to train user location prediction networks. To achieve accurate channel gain models, we design deep neural networks (DNNs) that are trained by signal measurements at each UAV. To cope with the challenging sequence-dependent characteristic of the problem, we develop a Lyapunov-based optimization framework with provable performance guarantees to decompose the original problem into a sequence of separate optimization subproblems based on the learned results. Finally, an alternative optimization scheme joint with a successive convex approximation technique is exploited to solve these subproblems. Simulation results demonstrate the accuracy of the learning methods as well as the effectiveness of the Lyapunov-based optimization framework.
\end{abstract}


%
\IEEEpeerreviewmaketitle

\section{Introduction}
%
%
%
%
\IEEEPARstart{5}{G} and emerging 6G wireless networks are expected to be highly agile and resilient and brace the capability of fast communication service recovery \cite{Gupta2016Survey} in case of network failure (e.g., infrastructure damage, flash crowd areas and remote areas). To achieve such an ambitious goal, the unmanned aerial vehicle (UAV) network has been considered as a significant component of 5G and 6G networks owing to its unique rapid response-ability and reduced vulnerability to natural disasters \cite{Cao2018Airborne}.

Meanwhile, 5G and 6G networks are convinced to accommodate different service requirements concerning communication latency, network throughput and communication reliability.
However, it is quite challenging to design a UAV network to satisfy these diverse service requirements simultaneously as UAVs have stringent size, weight and power consumption requirements. Fortunately, a UAV network can benefit from the network slicing characteristic of 5G and 6G networks via enabling virtually isolated on-board processing systems. In this way, multiple services can be converged onto a common UAV infrastructure, and the number of hardware components on UAVs can also be minimized, providing novel on-board system realizations \cite{garcia2019performance}.

\subsection{Prior works}
Recently, many research interests \cite{Sarah2018Autonomous,hellaoui2018aerial,garcia2019performance,xilouris2018uav,faraci2020design} have been paid to the UAV network slicing due to the significant role of the UAV network in 5G and 6G and the urgent requirement of improving the cost efficiency of providing diverse communication services by deploying UAVs.
For example,
the work in \cite{garcia2019performance} evaluated the performance of network slicing for UAV communications and demonstrated that the slicing was effective in terms of UAV payload slice and UAV control slice isolation.
A network slicing demo over 5G radio for UAV communications was also demonstrated in \cite{Sarah2018Autonomous}. In this demo, the network was virtually sliced into two slices where one slice was created for sending commands to a UAV, and the other slice was created to transmit payload from the UAV to a ground user.

As control information delivery has the stringent requirement of ultra-reliable and low-latency communications (URLLC) \cite{ren2019achievable} and payload needs to be transmitted over high-speed and broadband links \cite{abbasi2020trajectory}, two types of slices, i.e., URLLC slices for control information delivery and mobile broadband (MBB) slices for payload transmission, can be envisioned \cite{garcia2019performance} in the UAV network.
During the past few years, a rich body of works \cite{ren2019achievable,pan2019joint,ranjha2019quasi,wang2020packet,chen2020power,abbasi2020trajectory,ding20203d,wu2020optimal,abbasfar2019} on URLLC-enabled UAV network and MBB-enabled UAV network had been published.

In the research domain of URLLC-enabled UAV network, many works \cite{ren2019achievable,pan2019joint,ranjha2019quasi,wang2020packet,chen2020power} studied how to guarantee the performance of URLLC links for control and/or non-payload information transmission.
For instance, the average achievable data rate of the URLLC link between a ground control station and a UAV under a three-dimensional (3D) channel was studied in \cite{ren2019achievable}, where, the URLLC link was established to deliver control information for UAV collision avoidance.
Besides, to satisfy URLLC requirements of UAV control and non-payload links, a joint blocklength allocation and UAV location optimization with a goal of minimizing the decoding error probability under a low-latency constraint was researched in \cite{pan2019joint}.
In the research domain of MBB-enabled UAV network, most of the existing literatures \cite{abbasi2020trajectory,ding20203d,wu2020optimal,abbasfar2019} focused on the movement control and/or resource (e.g., transmit power, bandwidth) allocation of the UAV network towards MBB service coverage.
For example, 
a joint UAV trajectory planning and transmit power allocation for a UAV network was researched to extend the communication coverage for two disconnected far ground vehicles in \cite{abbasi2020trajectory}.
Additionally, the work in \cite{ding20203d} investigated the UAV trajectory design and bandwidth allocation problem considering both the UAV's energy consumption and the service fairness among the ground mobile users, and a deep reinforcement learning based algorithm was used to solve this problem. 

\subsection{Motivations and contributions}
However, to practically realize the vision of deploying the UAV network, the communication problem in the UAV network should be rationally formulated and effectively solved. To this aim, UAVs need to obtain location information of ground mobile users and require analytical channel gain models. Owing to such reasons, most of the above works \cite{ren2019achievable,pan2019joint,ranjha2019quasi,wang2020packet,chen2020power,abbasi2020trajectory,ding20203d,wu2020optimal,abbasfar2019} assumed that users' locations were known and adopted simplified channel gain models such as the isotropic radiation for antenna and the free space propagation model \cite{mudumbai2009medium} or complicated channel gain models like the probabilistic LoS model \cite{alhourani2014optimal} and angle-dependent channel parameters \cite{azari2018ultra}. The simplified models, however, may be inaccurate in practical environment as they do not have the slightest association with the local environment where UAVs are actually deployed. The model accuracy cannot be guaranteed in the local environment even exploiting sophisticated statistical models. This is because they can only simulate the channel gain in an average sense.

To tackle the user location-related issue, the work in \cite{chen2017caching} proposed to use a learning method to predict ground users' locations. Based on the predicted locations, UAVs were deployed to provide MBB services for users. However, this work still adopted the probabilistic LoS model to calculate UAV-to-ground-user (UtG) channel gain values.
To address the channel gain model-related issue, a recent work in \cite{zeng2020simultaneous} proposed to use a deep neural network (DNN), which was trained based on raw signal measurements, to obtain numerical channel gain estimation values.
However, this work did not derive analytically tractable channel gain models; as a result, no theoretical analysis on the communication problem in the UAV network could be conducted.

Additionally, slicing the UAV network should tackle the non-trivial mismatch issue between the slice supply and slice demand. The creation and configuration of UAV slices (slice supply), which require the protocol configuration and resource orchestration and release, are time-consuming; however, services (especially URLLC services) cannot tolerate the delay of creating and configuring the slices.

To overcome the above issues, we propose in this paper to proactively slice the UAV network for URLLC and MBB service multiplexing. Our main contributions are summarized as follows:
\begin{itemize}
\item Owing to the mobility of users and the limited UAV communication coverage, a time-varying UAV network is desired to be operated to improve mobile users' quality of service (QoS). Thus, we formulate the UAV network slicing problem as a sequential decision problem to provide energy-efficient and fair MBB services for ground mobile users while satisfying ultra-reliable and low-latency requirements of UAV control and non-payload information transmission.
This problem, however, is highly challenging to be solved via standard optimization methods mainly due to the lack of accurate users' locations and channel gain models, as well as the sequence-dependent characteristic.
\item A distributed learning method is exploited to predict users' locations as users' historical locations are scattered among UAVs and a UAV cannot rely solely on partial location information to predict users' locations. Besides, we propose to mitigate the mismatch issue between slice supply and slice demand by proactively slicing the UAV network.
Although the creation and configuration of slices are time-consuming, they can be performed proactively or in advance based on future users' locations.
\item We construct accurate and analytically tractable channel gain models based on estimation results of DNNs. This is because actual channel gain values depend on mobile users and flying UAVs in a rather sophistical manner and DNNs have the powerful non-linear function approximation ability.
\item Inspired by the superiority of the Lyapunov approach in tackling sequential decision problems, we propose to decompose the formulated problem into multiple repeated optimization subproblems based on the learned results via a Lyapunov-based optimization framework, which is provably performance guaranteed.
\item The subproblems are confirmed to be mixed-integer-non-convex, which are difficult to be solved. To make them tractable, an alternative optimization scheme and a successive convex approximation (SCA) technique are exploited to handle the mixed-integer and non-convexity properties, respectively. 
\item Finally, we conduct simulations to verify the accuracy of the learning methods and the effectiveness of the Lyapunov-based optimization framework as well.
\end{itemize}

\section{System Model}
As shown in Fig. \ref{fig_uav_slice_architecture}, this paper considers a communication coverage scenario by deploying a UAV network in urban environment\footnote{The proposed algorithm can be directly applied in many other types of environment even though the urban environment is considered here. Besides, although BSs are densely deployed in urban environment, it is still necessary to deploy a UAV network to recover partial communication coverage under emergency communication scenarios like flash crowd areas and terrestrial infrastructure malfunction areas.}. This scenario mainly includes a BS with array antenna configuration, $J$ single antenna UAVs, and $N^e$ pedestrians (or called ground mobile users) walking in a two-dimensional (2D) urban area of interest ${\mathbb R}^2$.
The BS is utilized to transmit ultra-reliable and low-latency control signals (e.g., UAV trajectories) to control the movement of UAVs and non-payload signals (e.g., UAV transmit power) to configure the UAV network through uplink wireless fading channels\footnote{The cooperation among UAVs can improve the UAV network resource utilization, we therefore study the case of controlling UAVs in a centralized manner. Yet, the communication scenario where BSs and UAVs cooperate to serve ground users is left for future research due to the space limitation. }.
The UAVs, the set of which is denoted by ${\mathcal J} = \{1, 2, \ldots, J\}$, acting as flying relays are deployed to perform a communication task, i.e., providing energy-efficient and fair service coverage for ground mobile users via downlink wireless fading channels.
\begin{figure}[!t]
\centering
\includegraphics[width=2.4in]{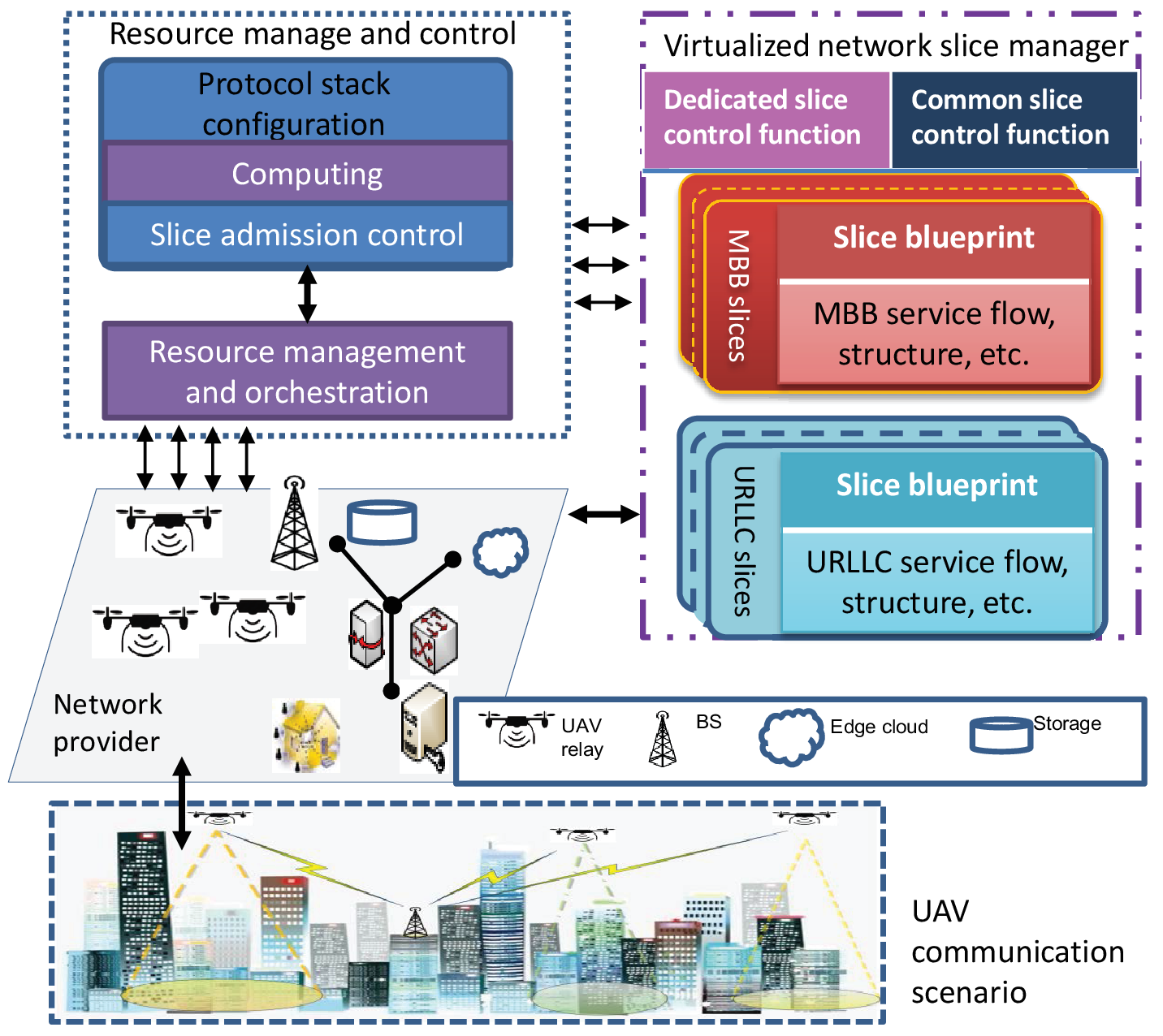}
\caption{A communication coverage scenario of a UAV network in urban environment and the UAV network slicing system architecture.}
\label{fig_uav_slice_architecture}
\end{figure}
Besides, owing to the limited number of UAVs and the UAV's restricted communication coverage, the movement of UAVs will be continuously controlled to complete the communication task. To theoretically model the communication task, we assume that the time domain in the considered scenario is discretized into a sequence of time slots and consider that the task may last long enough, i.e., $t = 1, 2, \ldots$, even though the working time of UAVs is limited mainly due to the energy constraint in practice.

\subsection{UAV network slicing system architecture}
In the above scenario, the UAV network has to simultaneously support ultra-reliable and low-latency uplinks for UAV control and non-payload signal transmission and energy-efficient downlinks for payload transmission. To achieve this goal, we propose
to virtually isolate UAV network resources and functions customized for specific requirements using the concept of network slicing.
Particularly, the UAV network is logically divided into two types of slices, i.e., URLLC slices and MBB slices\footnote{In practice, constrained by a UAV's service capability, it may be difficult to provide enhanced mobile broadband (eMBB) services, which are generally provided by terrestrial cellular networks, for many ground users using a UAV. We therefore do not investigate the eMBB use case.}.
A UAV network slicing system architecture is shown Fig. \ref{fig_uav_slice_architecture}. The system is composed of three major components: network provider (NP), resource manage and control (RMC), virtualized network slice manager (VNSM).
Following the 3GPP management reference framework \cite{5G20155G}, we consider the UAVs as a network infrastructure resource domain. NP owns the UAVs and the BS as well according to the business relationships and stakeholder roles defined in \cite{devlic2017nesmo}.
For each UAV, it can be considered as a node of network function virtualization (NFV) infrastructure. Based on the virtualized network function (VNF) and/or the physical network function (PNF) \cite{rost2017network}, user-specific radio resource control is activated at each UAV.
RMC is responsible for configuring radio access network protocol stacks according to service requirements of slices.
For example, for slices with high reliability and low-latency requirements, lower frame error rates, reduced round trip time (RTT), shortened transmission time interval (TTI), and/or multi-point diversity schemes are desired to be applied.
RMC is also in charge of orchestrating and releasing resources on request for all UAVs by exposing the northbound interface to the VNSM.
VNSM is made available by NP via logically abstracting the physical infrastructure resources as virtual computing, storage, and networking resources.
Besides, VNSM operating on the top of the physical and/or virtualized infrastructure is responsible for creating, activating, maintaining, configuring, and releasing slices during the life cycle of them. Through the dedicated/common slice control function, VNSM will generate a network slice blueprint (i.e., a template) for each accepted network slice. The slice blueprint describes the structure, configuration, control signals, and service flows for instantiating and controlling the network slice instance of a type of service during its life cycle. The slice instance includes a set of network functions and resources to meet the end-to-end service requirements. However, the above slice creation and configuration processes are time-consuming.


\subsection{URLLC slice model}
Before creating or activating a URLLC slice, the RMC must receive and admit a URLLC slice request. According to the network slice concept (from the QoS requirement viewpoint), the URLLC slice request is defined as below.
\begin{mydef}\label{urllc_slice_definition}
A URLLC slice request is characterized as a tuple $\{\tau_{s,u}^{req}, {\varepsilon _{s,u}^{req}}\}$ for a slice $s \in {\mathcal S}^u=\{1, 2, \ldots, |{\mathcal S}^u|\}$, where $\tau_{s,u}^{req}$ represents the transmission latency requirement of a data packet of control and non-payload signals in $s$, ${\varepsilon _{s,u}^{req}}$ denotes the codeword\footnote{A URLLC packet will usually be coded before transmission, and the generated codeword will be transmitted in the air interface such that the transmission reliability can be improved.} decoding error probability of the packet, $|{\mathcal S}^u|$ denotes the number of URLLC slices.
\end{mydef}

We assume that all URLLC slice requests can be accepted by the RMC. This assumption is rational because URLLC packets should be immediately served once arriving in the UAV network owing to the ultra-low-latency requirement. Even if all URLLC slice requests are accepted, URLLC slices will not occupy a lot of system resources as the number of UAVs is relatively small. Besides, UAVs that have the same reliable and low-latency requirement are served by the same URLLC slice, where ${\mathcal N}_s^u = \{1, 2, \ldots, N_s^u\}$ is the set of UAVs served by URLLC slice $s$ with $N_s^u$ being the number of UAVs served by slice $s \in {\mathcal S}^u$.

\subsection{MBB slice model}
Similarly, an MBB slice can only be created after the MBB slice request is accepted, and we define the MBB slice request as follows.
\begin{mydef}\label{MBB_slice_res_def}
An MBB slice request can be characterized as a tuple $\{I_s^e, C_s^{\rm th}\}$ for any slice $s \in {\mathcal S}^e = \{1, 2, \ldots, |{\mathcal S}^e|\}$, where $I_s^e$ is the number of served ground users in $s$, and $C_s^{\rm th}$ is the data rate requirement of each served user in $s$, $|{\mathcal S}^e|$ is the number of MBB slices.
\end{mydef}

As an MBB slice is created and configured to serve users with the same data rate requirement, users with different data rate requirements will be served by diverse slices.
Ground users are partitioned into $|{\mathcal S}^e|$ groups according to their data rate requirements in this paper.
Let ${\mathcal N}_s^e = \{1, 2, \ldots, N_s^e\}$ denote the set of users with the data rate requirement of $C_s^{\rm th}$ for all $s \in {\mathcal S}^e$.
However, owing to the resource limitation on the UAV network and $N^e = \sum\nolimits_{s \in {{\mathcal S}^e}} {N_s^e} $ is much greater than $J$, we do not assume that all MBB slice requests should be accepted.
To this end, for UAV $j \in {\mathcal J}$ and user $i \in {\mathcal N}_s^e$, we use an indicator variable $a_{ij,s}(t) \in {\mathcal A}(t)$ to indicate if the MBB slice request of creating an MBB slice $s$ to serve user $i$ using UAV $j$ at time slot $t$ can be accepted, where ${\mathcal A}(t)$ is an acceptance set of MBB slice requests. We let $a_{ij,s}(t) = 1$ if the slice request is accepted; otherwise, $a_{ij,s}(t) = 0$.


\subsection{Channel gain models}
The solution of the movement control problem of UAV requires the theoretical expressions of channel gain models.
In the considered scenario, we focus on the modelling of the BtU and UtG channel gain models.

\subsubsection{BtU channel gain model}
As mentioned above, some existing simplified and complicated statistical channel gain models \cite{mudumbai2009medium,alhourani2014optimal,azari2018ultra} are inaccurate in the local environment where UAVs are actually deployed.
To tackle this issue, for the BS and a typical UAV $j \in {\mathcal J}$ in URLLC slice $s \in {\mathcal S}^u$, we model the DNN-based BtU channel gain at time slot $t$, denoted by $h_{j,s}^{\rm B}(t)$, as follows:
\begin{equation}\label{eq:BtU_channel_gain}
h_{j,s}^{\rm B}(t) = G_{j,s}^{\rm B}(\bm x_{\rm B}^{\rm 3D}, \bm v_j^{\rm 3D}(t)) G_r \bar g_{j,s}^{\rm B}(\bm x_{\rm B}^{\rm 3D}, \bm v_j^{\rm 3D}(t)) f_{j,s}^{\rm B}(t) \buildrel \Delta \over = \frac{\theta_{j}^{\rm B}(t)}{{{D}_{j;{\rm B}}^2(t)}},
\end{equation}
where $G_{j,s}^{\rm B}(\bm x_{\rm B}^{\rm 3D}, \bm v_j^{\rm 3D}(t))$ denotes the transmitting antenna gain relying on BS antenna configurations such as the number of elements in antenna array and placement of these elements, $G_r$ is the receiving antenna gain, $\bm x_{\rm B}^{\rm 3D} = [x_{\rm B}, y_{\rm B}, g_{\rm B}]^{\rm T}$ is the 3D coordinate of the BS with $g_{\rm B}$ being the deployment altitude of the BS.
Besides, $\bar g_{j,s}^{\rm B}(\bm x_{\rm B}^{\rm 3D}, \bm v_j^{\rm 3D}(t))$ represents the path loss between the BS and UAV $j$ in $s$ at $t$, $f_{j,s}^{\rm B}(t)$ is a random variable denoting the small-scale fading, ${\theta_{j}^{\rm B}(t)} := \mu(\bm s^{\rm ul}(t)|\theta^{\mu_j} (t))$, where $\mu(\bm s^{\rm ul}(t)|\theta^{\mu_j} (t))$ is a DNN with $\bm s^{\rm ul}(t)$ being the DNN input and $\theta^{\mu_j} (t)$ being the DNN parameters, is a channel gain coefficient that will be obtained based on the DNN described in subsection V-A. ${{D}_{j;{\rm B}}(t)}= {{||{{\bm v}_j^{\rm 3D}}(t) - {{\bm x}_{\rm B}^{\rm 3D}}||_2}}$ is the distance between UAV $j$ and the BS at $t$.

The advantages of (\ref{eq:BtU_channel_gain}) are: the obtained channel gain can reflect the actual local environment parameters; the theoretical expression of the channel gain is not complicated, which will make the UAV communication problem analytically tractable.

\subsubsection{UtG channel gain model}
Similar to the definition of the BtU channel gain, for a typical UAV $j \in {\mathcal J}$ and a ground user $i \in {\mathcal N}_s^e$, $s \in {\mathcal S}^e$, we model the DNN-based UtG channel gain at time slot $t$ as follows:
\begin{equation}\label{eq:UtG_channel_gain}
h_{ij,s}(t) = G_j G_r\bar g_{ij,s}(\bm x_{i,s}^{\rm 3D}(t), \bm v_j^{\rm 3D}(t)) f_{ij,s}(t) \buildrel \Delta \over = \frac{\theta_{ij,s}(t)}{{{D}_{ij,s}^2(t)}},
\end{equation}
where $G_j$ denotes the transmitting antenna gain of UAV $j$, $\bar g_{ij,s}(\bm x_{i,s}^{\rm 3D}(t), \bm v_j^{\rm 3D}(t))$ represents the path loss between UAV $j$ and user $i$ at time slot $t$, $f_{ij,s}(t)$ is a random variable accounting for the small-scale fading.
$\bm x_{i,s}^{\rm 3D}(t) = [x_{i,s}(t), y_{i,s}(t), g_{i,s}]^{\rm T}$ represents the 3D coordinates of user $i$ with $g_{i,s}$ being the height of user $i$.
$\bm v_j^{\rm 3D}(t) = [x_{j}(t), y_{j}(t), g_j(t)]^{\rm T}$ is the 3D coordinates of UAV $j$ at $t$ with $g_j(t)$ being the deployment altitude of UAV $j$.
$\theta_{ij,s}(t) := Q(\bm s^{\rm dl}(t)|\theta^{Q_j}(t))$, where $Q(\bm s^{\rm dl}(t)|\theta^{Q_j}(t))$ denotes a DNN with $\bm s^{\rm dl}(t)$ being the DNN input and $\theta^{Q_j}(t)$ being the DNN parameters, is the channel gain coefficient that will be determined based on the DNN described in subsection V-B.
${{D}_{ij,s}(t)}= {{||{{\bm v}_j^{\rm 3D}}(t) - {{\bm x}_{i,s}^{\rm 3D}(t)}||_2}}$ is the distance between UAV $j$ and user $i$ at $t$.


Additionally, we denote $\bm x_{i,s}(t) = [x_{i,s}(t), y_{i,s}(t)]^{\rm T}$ as the horizontal location of user $i$ at $t$ and denote the horizontal location of the BS by $\bm x_{\rm B} = [x_{\rm B}, y_{\rm B}]^{\rm T}$.
We consider the case that all UAVs are deployed at the same and fixed altitude in this paper and leave the movement control problem of UAVs in 3D space for future research.
The horizontal location of UAV $j$ is denoted by $\bm v_j(t) = [x_{j}(t), y_{j}(t)]^{\rm T} \in {{{\cal X}}}(t)$, where ${{{\cal X}}}(t)$ is the set of UAVs' horizontal locations at $t$.

\section{Problem Formulation}
In this section, we formulate the problem of slicing a UAV network to provide energy-efficient and fair communication coverage for ground mobile users as an optimization problem. To this aim, we first enforce intra-slice constraints and inter-slice constraints and then define the objective function of the optimization problem. With aforementioned system model, constraints and the objective function, the optimization problem is formulated.

\subsection{URLLC slice constraints}
To support URLLC services, 
the constraints on some crucial performance indicators like transmit data rates of transmitters in a network should be rigorously satisfied.
In wireless communication networks, Shannon formula is usually leveraged to quantize the transmit data rate of a transmitter. However, in the network supporting the URLLC transmission, Shannon formula cannot be utilized. This is because Shannon formula is estimated under the crucial assumption of transmitting packets of enough long blocklength; yet, the length of a URLLC packet is typically very short to satisfy the stringent low-latency requirement. To tackle this issue, like \cite{Yang2020Joint,Ren2020Joint}, we assume that the fading channel is a quasi-static Rayleigh fading channel over a time slot and changes independently. Then, the rate formula in finite blocklength regime in \cite{Ren2020Joint} is exploited to approximate the achievable data rate of transmitting URLLC packets from the BS to UAV $j$ in URLLC slice $s \in {\mathcal S}^u$ under the given transmission latency $\tau_{s,u}^{\rm req}$ and codeword decoding error probability $\varepsilon _{s,u}^{\rm req}$ \cite{Ren2020Joint}, i.e.,
\begin{equation}\label{rate_urllc}
R_{j,s}^u(t) \approx \frac{{W_{j,s}^u(t)}}{{\ln 2}}\left[ {\ln \left( {1 + \frac{{{h_{j,s}^{\rm B}(t)}p_{j,s}^{\rm B}(t)}}{{{N_0}W_{j,s}^u(t)}}} \right) - \sqrt {\frac{{{V_{j,s}(t)}}}{{\tau_{s,u}^{\rm req} W_{j,s}^u(t)}}} {Q^{ - 1}}(\varepsilon _{s,u} ^{\rm req} )} \right],
\end{equation}
where $p_{j,s}^{\rm B}(t)$ is the transmit power of the BS when connecting to UAV $j$ at slot $t$, $W_{j,s}^u(t)$ is the system bandwidth allocated to UAV $j$ in URLLC slice $s$,
$N_0$ is the noise power spectral density, $Q^{ - 1}\left(  \cdot  \right)$ is the inverse of Q-function, and ${V_{j,s}(t)} = 1 - {1 \mathord{\left/
 {\vphantom {1 {{{\left( {1 + \frac{{{h_{j,s}^{\rm B}(t)}p_j^{\rm B}(t)}}{{{N_0}W_{j,s}^u(t)}}} \right)}^2}}}} \right.
 \kern-\nulldelimiterspace} {{{\left( {1 + \frac{{h_{j,s}^{\rm B}(t)p_{j,s}^{\rm B}(t)}}{{ {N_0}W_{j,s}^u(t)}}} \right)}^2}}}$ is the channel dispersion. {Note that, to ensure the low-latency requirements of UAVs, a frequency division multiple access (FDMA) technique is applied to achieve the URLLC inter-slices and intra-slices isolation.}

When the received signal-to-noise ratio (SNR) is higher than 5 dB, ${{V_{j,s}(t)}}$ can be very accurately approximated as $1$ \cite{schiessl2015delay}. On the other hand, even in low SNR regime, since ${{V_{j,s}(t)}} < 1$, we can obtain the upper bound of the minimum required transmit power by substituting ${{V_{j,s}(t)}} = 1$ into (\ref{rate_urllc}). If the upper bound value is applied in optimizing resource allocation, then the requirements on transmission latency, codeword decoding error probability in the URLLC communication can be satisfied. Besides, to satisfy the low-latency requirement, the minimum data rate is $R_{j,s}^u(t) = b_{s,u}^{\rm req} / \tau_{s,u}^{\rm req}$, where $b_{s,u}^{\rm req}$ denotes the number of bits to be transmitted within $\tau_{s,u}^{\rm req}$. Then, by activating the transmit data rate condition and approximating $V_{j,s}(t) = 1$, the mathematical expression of the required transmit power from the BS to UAV $j$ that satisfies the transmission latency and codeword decoding error probability requirements can be derived as
\begin{equation}\label{power_urllc_app}
p_{j,s}^{\rm B}(t) = \frac{{{N_0}W_{j,s}^u(t)}}{{h_{j,s}^{\rm B}(t)}}\left\{ {\exp \left[ {\frac{{{b_{s,u}^{\rm req}}\ln 2}}{{{\tau _u^{\rm req}}W_{j,s}^u(t)}} + \frac{{Q^{ - 1}\left( {{\varepsilon _{s,u}^{\rm req}}} \right)}}{{\sqrt {{\tau _{s,u}^{\rm req}}W_{j,s}^u(t)} }}} \right] - 1} \right\}.
\end{equation}

Let $p({W^u(t)})$ denote the total required transmit power of the BS for connecting to all UAVs with given $W^u(t) = \sum_{s,j} W_{j,s}^u(t)$, where we lighten the notation $\sum_{s,j} W_{j,s}^u(t)$ for $\sum_{s \in {\mathcal S}^u}\sum_{j\in{\mathcal N}_s^u} W_{j,s}^u(t)$. The similar lightened notation is adopted throughout the rest of the paper to simplify the description.
Since the maximum transmit power denoted by $p_{\rm B}^{\rm max}$ of the BS is limited, we have the following transmit power constraint
\begin{equation}\label{con_urllc_power}
p({W^u(t)}) = \sum\nolimits_{s,j} {p_{j,s}^{\rm B}(t)} \le p_{\rm B}^{\rm max}.
\end{equation}

\subsection{MBB slice constraints}
Owing to the movement of UAVs, user $i$ for all $i \in {\mathcal I}_s^e$, $s \in {\mathcal S}^e$ may be in the communication ranges of several UAVs at slot $t$.
We assume that at $t$, a user can be served by at most one UAV, and a UAV is allowed to deliver MBB traffic to at most one user due to its limited service capability.
In this way, the upper layer network slice configuration (e.g., protocol stack configuration, slice blueprint generation), which is time-consuming, can be proactively identified by RMC and VNSM to accommodate the data rate requirements of admitted users \cite{ni2019end}. Mathematically, we have
\begin{equation}\label{eq:slice_request_indicator}
\sum\nolimits_{j \in {\mathcal J}} {{a_{ij,s}}(t)}  \le 1,\forall i \in {\mathcal N}_s^e, s \in {\mathcal S}^e, \text{ } \sum\nolimits_{i,s}{{{a_{ij,s}}(t)}}  \le 1,\forall j \in {\mathcal J}
\end{equation}

Based on the above slice request acceptance condition, the system bandwidth, denoted by $W^e(t)$, allocated to MBB slices at $t$ can be fully reused by each ground mobile user. In this case, for a user $i \in {\mathcal N}_s^e$, $s \in {\mathcal S}^e$, we denote its received signal-to-interference-plus-noise ratio from UAV $j$ at time slot $t$ by ${\rm SINR}_{ij,s}(t)$, which can be expressed as
\begin{equation}\label{eq:MBB_sinr}
{{\rm SINR}_{ij,s}}(t) = \frac{{{p_j}(t){h_{ij,s}}(t)}}{{N_0 W^e(t) + {I_{ij,s}}(t)}},
\end{equation}
where $p_j(t) \in {\mathcal P}(t)$ is the instantaneous transmit power of UAV $j$ at $t$ with ${\mathcal P}(t)$ being a set of possible values of all UAVs' transmit power, ${I_{ij,s}}(t) = \sum\nolimits_{k \in {\mathcal J}\backslash \{ {j}\} } {{p_k}(t){h_{ik,s}}(t)} $ is the interference caused by other UAVs,




The time average transmit power of UAV $j$ during the first $t$ time slots can be written as ${\bar p_j}(t) = \frac{1}{t}\sum\nolimits_{\tau  = 1}^t {{p_j}(\tau )}$. Except for the transmit power, UAVs are subject to inherent circuit power consumption mainly including power consumption of mixers, frequency synthesizers, and digital-to-analog converters. Denote $p_{j}^{c}$ as the circuit power of $j$ during a time slot, we then model the total power consumption of $j$ at $t$ as
\begin{equation}\label{eq:power_total_at_t}
p_j^{\rm tot}(t) = {p_j}(t) + p_j^c,
\end{equation}
which is upper-bounded by the maximum instantaneous total power ${\hat p_j}$, i.e., $p_j^{tot}(t) \le {\hat p_j}$.
Accordingly, the time average total power consumption of UAV $j$ during the first $t$ time slots can be written as
\begin{equation}\label{eq:8}
\bar p_j^{\rm tot}(t) = {\bar p_j}(t) + p_j^c,
\end{equation}
which is constrained by $\bar p_j^{\rm tot}(t) \le {\tilde p_j}$, and ${\tilde p_j}$ is the maximum time average total power consumption of UAV $j$.



We then leverage Shannon formula to quantify the achievable data rate $u_{i,s}(t)$ (in Mbps) of user $i \in {\mathcal I}_s^e$, $s\in {\mathcal S}^e$ at $t$ as given by
\begin{equation}\label{eq:achievable_data_rate}
{u_{i,s}}(t) = \sum\nolimits_{j \in {\mathcal J}} {{a_{ij,s}}(t)W^e(t){{\log }_2}\left( {1 + {{\rm SINR} _{ij,s}}(t)} \right)}.
\end{equation}

During the first $t$ time slots, the time average achievable data rate of user $i$ can be written as ${\bar u_{i,s}}(t) = \frac{1}{t}\sum\nolimits_{\tau  = 1}^t {{u_{i,s}}(\tau )}$.
As users require the minimum time average achievable data rates in practical communication scenarios, we present a constraint to guarantee that user $i$'s minimum data rate requirement is satisfied, i.e.,
\begin{equation}\label{eq:average_rate_cons}
{{\bar u}_{i,s}}(t) \ge C_s^{\rm th}, \forall i \in {\mathcal N}_s^e, s\in {\mathcal S}^e.
\end{equation}

\subsection{Physical resource and UAV movement constraints}
During the flight, the distance between two consecutive waypoints on a UAV trajectory will be constrained by the UAV's maximum speed. As such, the mathematical expression of the waypoint distance constraint can be written as
$ {|| {{{\bm v}_j}(t) - {{\bm v}_j}(t - 1)} ||_2^2} \le e_{\max }^2$, where $e_{\max }$ is the UAV's maximum flight distance during a time slot.
Additionally, for collision avoidance, the distance between any two UAVs at each slot should not be less than a safety distance. Mathematically, the expression can be written as ${|| {{{\bm v }_j}(t) - {{\bm v}_k}(t)} ||_2^2} \ge d_{\min }^2$, where $d_{\min }$ is the minimum safety distance.

Besides, since the MBB and the URLLC service provisions are considered and network bandwidth resources allocated to MBB and URLLC slices are separated in the frequency plane to achieve the inter-slice isolation, the network bandwidth constraint can be written as
\begin{equation}\label{total_bandwidth}
{W^u}(t) + {W^{e}(t)} = W^{\rm tot},
\end{equation}
where $W^{\rm tot}$ denotes the total system bandwidth.

\subsection{Objective function and problem formuation}
Define $\phi ( \{{{\bar{u}}_{i,s}}(t)\} )={\sum\nolimits_{i,s}{{{\log }_{2}}(1+{{{\bar{u}}}_{i,s}}}(t))}$ as a proportional fairness function of time average achievable data rates across all ground users. The maximization of $\phi (\{{{\bar{u}}_{i,s}}(t)\})$ will lead to that of users' time average achievable data rates as well as UAVs' fair coverage.
Our goal is to achieve an energy-efficient and fair MBB service provision by physically configuring the UAV {network} including optimizing the acceptance set of MBB slice requests ${\mathcal A}(t)$, controlling UAVs' movement ${\mathcal X}(t)$, and optimizing UAV transmit power ${\mathcal P}(t)$ during the whole period of the communication task. Combining with the above analysis, we can formulate the UAV network slicing problem as a sequential decision problem presented as below
\begin{subequations}\label{eq:original_problem}
\begin{alignat}{2}
& \mathop {\rm Maximize }\limits_{{\mathcal A}(t),{\mathcal P}(t),{{{\mathcal X}}}(t)} {\mkern 1mu} \mathop {\lim \inf }\limits_{t \to \infty } ( {\phi ( \{{{\bar{u}}_{i,s}}(t)\} ) - \rho \sum\nolimits_{j \in {\mathcal J}} {\bar p_j^{\rm tot}(t)} } ) \\
&{\rm s.t:} \text{ } \mathop {\lim \inf }\limits_{t \to \infty } {{\bar u}_{i,s}}(t) \ge C_s^{\rm th},\forall i,s \\
& \quad \mathop {\lim \sup }\limits_{t \to \infty } \bar p_j^{\rm tot}(t) \le {{\tilde p}_j},\forall j \\
& \quad  p_j^{\rm tot}(t) \le {{\hat p}_j},\forall j,t \\
& \quad {\left\| {{{\bm v}_j}(t) - {{\bm v}_j}(t - 1)} \right\|^2} \le e_{\max }^2{\rm{,}}\forall j,t \\
& \quad {\left\| {{{\bm v}_j}(t) - {{\bm v}_k}(t)} \right\|^2} \ge d_{\min }^2{\rm{,}}\forall j,k \ne j,t \\
& \quad {a_{ij,s}}(t) \in \{ 0,1\} ,\forall i,s,j,t \\
& \quad W_{j,s}^u(t) > 0, {\rm{ }}\forall j,s,t \\
& \quad (\ref{con_urllc_power}),(\ref{eq:slice_request_indicator}),(\ref{total_bandwidth}),
\end{alignat}
\end{subequations}
where {${{\bm v}_{j}}(0)$ represents the initial horizontal location of UAV $j$}, $\rho$ is a non-negative coefficient that weighs a trade-off between the system revenue and the power consumption.

However, the solution of (\ref{eq:original_problem}) is highly challenging. This is because the locations of ground mobile users are unknown at each time slot and analytically tractable channel gain models are not obtained.
Besides, the sequence-dependent characteristic of (\ref{eq:original_problem}) significantly hinders its solution.
To solve such a challenging problem, we first propose to predict ground mobile users' locations using a distributed learning method. We then design an online learning method to estimate channel gain coefficients. Based on the predicted users' locations and the estimated channel gain coefficients, we construct analytically tractable channel gain models. Finally, with the predicted users' locations and constructed channel gain models, we exploit a Lyapunov-based optimization framework to solve the sequential decision problem. The procedures of solving (\ref{eq:original_problem}) are elaborated in the following three sections.

\section{Distributed learning for users' location prediction}
Since users' locations may continuously change as time elapses, dynamic slice creation and configuration should be enabled to improve the QoS of users. Considering the time-consuming slice creation and configuration,
(\ref{eq:original_problem}) should be proactively solved on the basis of predicted users' locations to achieve that goal.
To predict users' locations, it is necessary to apply a machine learning method.
Yet, a user's location information may be scattered in multiple UAVs owing to the movement of UAVs, and each UAV can only collect partial information of users' locations after a period of time. A UAV, however, cannot accurately predict users' locations based on partial location information.
Thus, some centralized methods performed at the BS may be exploited to accurately predict users' locations. However,
they require a large amount of raw user location-related data exchange among the BS and all UAVs, which will consume lots of network resources.
To tackle this issue, we develop a distributed learning method. In this method, although each UAV will learn locally, it can obtain the global prediction model by exchanging users' location prediction models rather than large amounts of raw data between it and the BS. Moreover, owing to the local learning and the exchange of prediction models, the BS and UAVs can predict each user's locations. With the predicted users' locations, the BS can effectively control the movement of UAVs via solving (\ref{eq:original_problem}). UAVs can construct UtG channel gain models with the predicted users' locations.


\subsection{Components of echo state network}
In this paper, an echo state network (ESN)-based learning method is exploited to train a location prediction model because the training process of the ESN is simple and fast and the ESN can effectively perform sequence-dependent data mining \cite{scardapane2016a}.
An ESN is a recurrent neural network which can be partitioned into four components: agent, input, ESN model, and output as specified below:
\begin{itemize}
\item \textbf{Agent:} There are $J+1$ agents separately located on $J$ UAVs and the BS as shown in Fig. \ref{fig_Distributed_ESN_architecture}. Each UAV agent will train ESN models locally and then send trained local ESN models to the BS agent for aggregation.
For the BS agent, it will aggregate received local ESN models and then broadcast the aggregated (or global) ESN model to all UAV agents.
\item \textbf{Input:} For a typical UAV $j$ and user $i \in {\mathcal N}_s^e$, $s \in {\mathcal S}^e$, the input set of the ESN is defined as ${\bm V_{ij,s}}(t) = \{{\bm x_{i,s}}(t - Q),...,{\bm x_{i,s}}(t)\}$ where $Q$ is the number of users' location samples. The $Q$ location samples $\{{\bm x_{i,s}}(t - Q),...,{\bm x_{i,s}}(t-1)\}$ are used to train an ESN model related to user $i$. ${\bm x_{i,s}}(t)$ is used to predict user $i$'s location.
\item \textbf{ESN model:} For a typical UAV $j$, one of its local ESN models is leveraged to correlate the $\bm x_{i,s}(t)$ with user $i$'s predicted location. As a single ESN, which has one hidden layer as shown in Fig. \ref{fig_Distributed_ESN_architecture} connecting the input and the output, can quickly converge, we regard it as a location prediction model. Consequently, a local ESN model includes an input weight matrix $\bm W_{\rm in}^r$, a recurrent matrix $\bm W_{\rm r}^r$ and an output weight matrix $\bm W_{j,t}$. 
\item \textbf{Output:} For a typical UAV $j$ and user $i \in {\mathcal N}_s^e$, $s \in {\mathcal S}^e$, the output of the ESN is defined as a matrix $\bm {\hat Y}_{ij,s}(t) = [\bm {\hat y}_{{i,s}}(t+1),\ldots,\bm {\hat y}_{i,s}(t+K)]$ where $K$ is the number of future time slots. $\bm {\hat y}_{i,s}(t+k)$ is the predicted location of user $i$ at time slot $t+k$. For the BS agent, the predicted output matrix is denoted by $\bm {\hat Y}_{{i,s}}^{\rm B}(t) = [\bm {\hat y}_{i,s}^{\rm B}(t+1),\ldots,\bm {\hat y}_{i,s}^{\rm B}(t+K)]$ for all $i \in {\mathcal N}_s^e$, $s \in {\mathcal S}^e$.
\end{itemize}

\subsection{Distributed ESN learning for users' location prediction}
In this subsection, the procedure of training all the local ESN models in a distributed way and then forcing these models to the global model at convergence is presented in detail.
Fig. \ref{fig_Distributed_ESN_architecture} shows the training processes of all UAVs and the interrelationship between UAVs and the BS.
\begin{figure}[!t]
\centering
\includegraphics[width=2.8in]{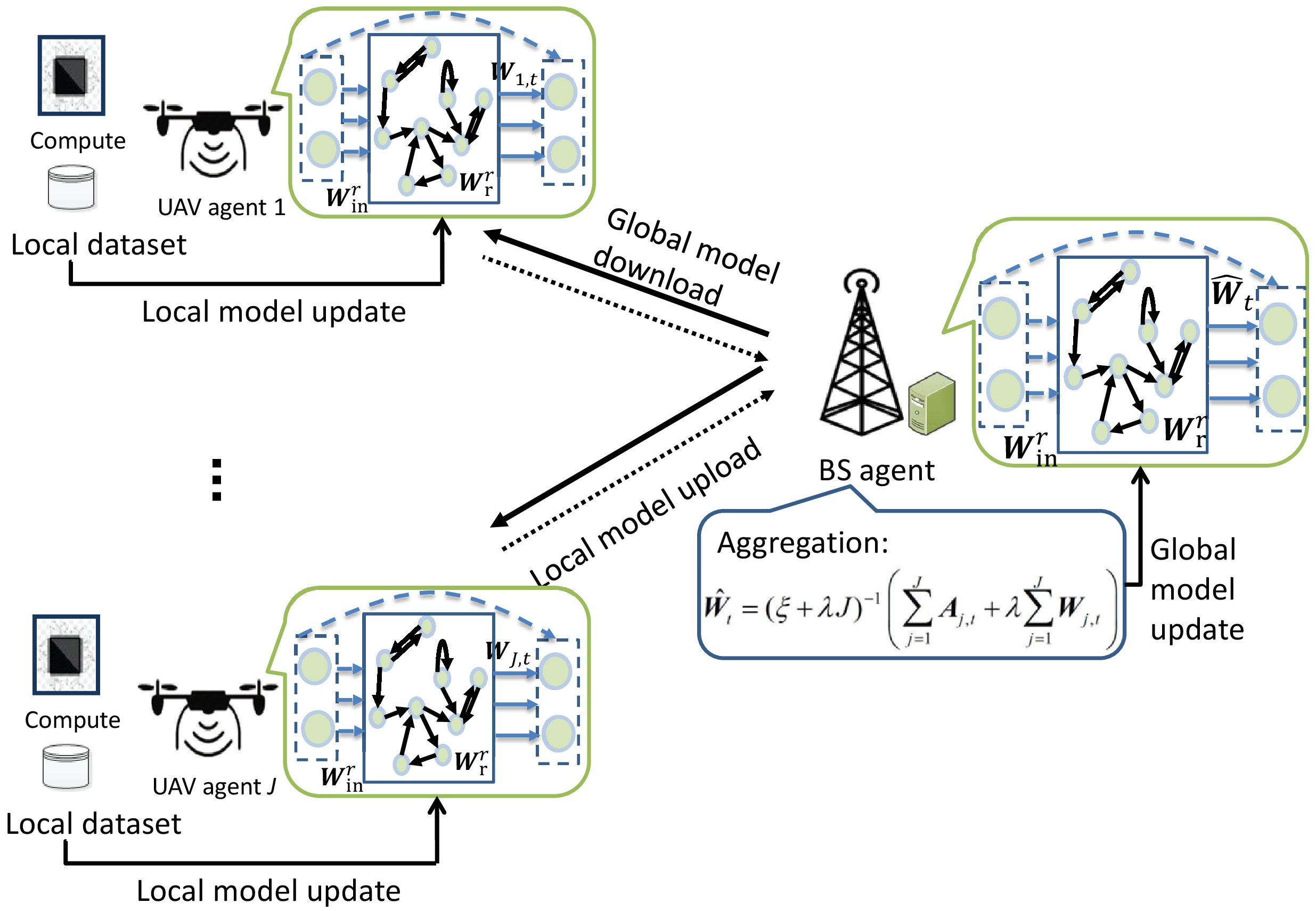}
\caption{The training processes of all UAVs and interrelationships between UAVs and the BS.}
\label{fig_Distributed_ESN_architecture}
\end{figure}

For user $i \in {\mathcal N}_s^e$, $s \in {\mathcal S}^e$, an input vector $\bm x_{i,s}(t) \in {\mathbb R}^{N_i}$, where $N_i$ represents the dimension of the vector, is fed to a reservoir with a dimension of $N_r$, whose internal state $\bm q_{i,s}(t-1) \in {\mathbb R}^{N_r}$ is updated as follows
\begin{equation}\label{eq:state}
\bm q_{i,s}(t) = f_{res}\left( \bm W_{\rm in}^r \bm x_{i,s}(t) + \bm W_{\rm r}^r \bm q_{i,s}(t-1) \right),
\end{equation}
where $\bm W_{\rm in}^r \in {\mathbb R}^{N_r \times N_i}$ and $\bm W_{\rm r}^r \in {\mathbb R}^{N_r \times N_r}$ are random matrices uniformly distributed in the interval $(0, 1)$, $f_{res}$ is a suitably defined non-linear function (e.g., $tanh(\cdot)$).

The predicted output of the ESN at $t$ is given by
\begin{equation}\label{eq:output}
{\hat {\bm y}_{i,s}}(t+1) = \bm W_{\rm in}^o \bm x_{i,s}(t) + \bm W_{\rm r}^o \bm q_{i,s}(t),
\end{equation}
where $\bm W_{\rm in}^o \in {\mathbb R}^{N_o \times N_i}$, $\bm W_{\rm r}^o \in {\mathbb R}^{N_o \times N_r}$ are trained based on the location samples. 

To train local ESN models on UAV $j$, we should be provided with a sequence of $Q$ target input-output pairs $\{(\bm x_{i,s}(t-Q), {\bm y}_{i,s}(t-Q+1)), \ldots, (\bm x_{i,s}(t-1), {\bm y}_{i,s}(t))\}$ with ${\bm y}_{i,s}(t)$ representing the real 2D position of user $i \in {\mathcal N}_s^e$, $s\in {\mathcal S}^e$, at $t$.
Besides, training ESN models requires updated users' locations. To this end, each user will broadcast Beacon messages containing location and time information every $T_p$ time slots to refresh its location stored on UAVs that can receive the Beacon messages.
Define the hidden matrix $\bm X_{ij,st}$ as
\begin{equation}\label{eq:hidden_matrix}
\bm X_{ij,st} = \left[ {\begin{array}{*{20}{c}}
{{\bm x_{i,s}^{\rm T}}(t-1){\bm q_{i,s}^{\rm T}}(t-1)}\\
 \vdots \\
{{\bm x_{i,s}^{\rm T}}(t-Q){\bm q_{i,s}^{\rm T}}(t-Q)}
\end{array}} \right].
\end{equation}

The optimal output weight matrix at time slot $t$ is then achieved by solving the regularized least-square problem:
\begin{equation}\label{eq:least_square_prob}
\bm W_t^{\star} = \mathop {\arg \min }\limits_{\bm W_t \in {R^{({N_i} + {N_r}) \times {N_o}}}} \frac{1}{2} \sum_{j=1}^{J}|| {\bm X_{ij,st} \bm W_t - \bm Y_{ij,st}} ||^2 + \frac{\xi}{2} || \bm W_t ||^2,
\end{equation}
where $\bm W_t = [\bm W_{\rm in}^o, \bm W_r^o]^{\rm T}$, $\xi \in {\mathbb R}_{+}$ is a positive scalar known as regularization factor, $\bm Y_{ij,st} = [\bm y_{i,s}(t)^{\rm T};\ldots;\bm y_{i,s}(t-Q+1)^{\rm T}] \in {\mathbb R}^{Q \times N_o}$.

As $\{\bm X_{ij,st}\}$ and $\{\bm Y_{ij,st}\}$ for all $i$, $j$, $s$, $t$ are locally collected, we adopt an alternating direction method of multipliers (ADMM) \cite{boyd2011distributed} to address (\ref{eq:least_square_prob}). This is because the ADMM method is a distributed method and it can quickly aggregate the locally obtained results to a global result at the convergence.

Particularly, by enforcing $\bm W_{j,t} = \hat {\bm W}_t$, we can obtain the augmented Lagrangian function of (\ref{eq:least_square_prob})
\begin{equation}\label{eq:lagrangian_ADMM}
\begin{array}{l}
{\mathcal L}(\bm W_{j,t}, \hat {\bm W}_t) = \frac{1}{2}\sum\limits_{j = 1}^J {||{\bm X_{ij,st}}{\bm W_{j,t}} - {\bm Y_{ij,st}}|{|^2}}  + \frac{\xi }{2}||{\hat {\bm W}}_t|{|^2} +
\sum\limits_{j = 1}^J {\rm tr}\left({{\bm A_{j,t}^{\rm T}} \left( {{\bm W_{j,t}} - \hat {\bm W}_t} \right)}\right)  + \frac{\lambda }{2}\sum\limits_{j = 1}^J {||{{\bm W}_{j,t}} - \hat {\bm W}_t|{|^2}},
\end{array}
\end{equation}
where ${\bm A_{j,t}} \in {\mathbb R}^{(N_i + N_r) \times N_o}$ is a Lagrangian multiplier matrix. $\{\bm W_{j,t}\}$ can be considered as a family of local variables, and $\hat {\bm W}_t$ can be regarded as the global consensus variable, $\lambda$ is a Lagrangian multiplier.

Next, an iterative framework is developed to mitigate (\ref{eq:lagrangian_ADMM}). By zero-forcing the derivative of ${\mathcal L}(\bm W_{j,t}, \hat {\bm W}_t)$ over $\bm W_{j,t}$ and $\hat {\bm W}_t$ at the $r$-th iteration, we can obtain the evolutionary forms of $\bm W_{j,t}^{(r)}$ and $\hat {\bm W}_t^{(r)}$, respectively,
\begin{equation}\label{eq:evolution_wj}
\begin{array}{l}
\bm W_{j,t}^{(r+1)} = {\left( {\bm X_{ij,st}^{\rm T}{\bm X_{ij,st}} + \lambda \bm I} \right)^{ - 1}}\left( {\bm X_{ij,st}^{\rm T}{\bm Y_{ij,st}} + \lambda \hat {\bm W}_t^{(r)} - {\bm A_{j,t}^{(r)}}} \right),
\end{array}
\end{equation}
where $\bm I \in {\mathcal Z}^{(N_i + N_r) \times (N_i + N_r)}$ is an identity matrix, and
\begin{equation}\label{eq:evolution_hatw}
\begin{array}{l}
\hat {\bm W}_t^{(r+1)}= {(\xi  + \lambda J)^{ - 1}} \left( {\sum\nolimits_{j = 1}^J {\bm A_{j,t}^{(r)}}  + \lambda \sum\nolimits_{j = 1}^J {\bm W_{j,t}^{(r+1)}} } \right).
\end{array}
\end{equation}

To update $\bm A_{j,t}^{(r+1)}$, the gradient ascend method is leveraged with
\begin{equation}\label{eq:lagrangian_aj}
\bm A_{j,t}^{(r+1)} = {\bm A}_{j,t}^{(r)} + \eta \left( {\bm W}_{j,t}^{(r+1)} - \hat {\bm W}_t^{(r+1)} \right),
\end{equation}
where $\eta$ is a constant step.

By exploring the ADMM, the distributed ESN learning method can be implemented for users' location predication. Summarily, $J$ UAVs will separately calculate (\ref{eq:evolution_wj}) and (\ref{eq:lagrangian_aj}). The BS calls (\ref{eq:evolution_hatw}) to aggregate the local variables $\{\bm W_{j,t}\}$. The local Lagrangian multipliers $\{\bm A_{j,t}\}$ are updated to drive $\{\bm W_{j,t}\}$ into their global consensus. After a limited number of iterations, $\{\bm W_{j,t}\}$ will converge to the global consensus variable $\bm W_t^{\star} = \hat {\bm W}_t$.
With the obtained $\bm W_t^{\star}$, each UAV and the BS will be able to predict the entire locations of all users. The main steps of predicting users' locations are given in Algorithm \ref{alg_user_loc_prediction}.
\begin{algorithm}[!htp]
\caption{Distributed ESN learning for users' location prediction}
\label{alg_user_loc_prediction}
\begin{algorithmic}[1]
\STATE {\textbf {Input:}} Training data set (local), $\{\bm V_{ij,s}(t)\}$ $\forall i \in {\mathcal N}_s^e$, $s\in {\mathcal S}^e$, $j \in {\mathcal J}$, the maximum number of iterations $r_{\rm max}$.
\STATE {\textbf {Output:}} The predicted location $\bm {\hat Y}_{i,s}^{\rm B}(t)$, $\bm {\hat Y}_{ij,s}(t)$ for all $i$, $s$. 
\STATE Each UAV agent $j$ generates $N^e$ local ESN models including $N^e$ matrices $\bm W_{\rm in}^r$ and $N^e$ matrices $\bm W_{\rm r}^r$ for predicting $N^e$ users' locations, respectively.
\STATE Each UAV agent $j$ gathers hidden matrices $\{\bm X_{ij,st}\}$ and teacher data $\{\bm Y_{ij,st}\}$ from $\{\bm V_{ij,s}(t)\}$ for all $i$, $s$.
\STATE Initialize $\bm A_{j,t}^{(1)} = \bm 0$ and $\hat {\bm W}_{t}^{(1)} = \bm 0$.
\FOR {$r = 1:r_{\rm max}$}
\FOR {Each UAV agent $j$ in parallel}
\STATE Agent $j$ computes $\bm W_{j,t}^{(r+1)}$ using (\ref{eq:evolution_wj}) and transmits $\bm W_{j,t}^{(r+1)}$ to the BS agent.
\ENDFOR
\STATE The BS agent computes $\hat {\bm W}_t^{(r+1)}$ using (\ref{eq:evolution_hatw}) and broadcasts $\hat {\bm W}_t^{(r+1)}$ to all UAV agents.
\FOR {Each UAV agent $j $ in parallel}
\STATE Agent $j$ computes $\bm A_{j,t}^{(r+1)}$ using (\ref{eq:lagrangian_aj}) and transmits $\bm A_{j,t}^{(r+1)}$ to the BS agent.
\ENDFOR
\STATE {Break if it converges or $r = r_{\rm max}$}
\ENDFOR
\STATE The BS agent can then obtain user $i$'s, $\forall i$, predicted locations $\bm {\hat Y}_{i,s}^{\rm B}(t)$ by iteratively assigning $\bm x_{i,s}(t+1) = \bm {\hat y}_{i,s}^{\rm B}(t)$ and calling (\ref{eq:output}) for $K$ times. Similarly, each UAV can achieve all users' predicted locations.
\end{algorithmic}
\end{algorithm}

\subsection{Convergence of the distributed ESN learning method}
In this subsection, we will analyze the convergence of Algorithm \ref{alg_user_loc_prediction}. The following lemma shows that Algorithm \ref{alg_user_loc_prediction} is convergent.
\begin{lemma}\label{lem:ADMM_ESN}
    Denote $(\bm W_{j,t}^{\star}, \hat {\bm W}_t^{\star})$ as the optimal solutions to (\ref{eq:lagrangian_ADMM}). For the distributed ESN learning method, $\forall r \in \mathbb{Z}^+$, $j \in {\mathcal J}$, we have that ${{\mathcal L}(\bm W_{j,t}^{(r)}, \hat {\bm W}_t^{(r)})} $ is bounded and
    \begin{equation}\label{eq:Lagrangian_drift}
         {\mathcal L} (\bm W_{j,t}^{\star}, \hat {\bm W}_t^{\star}) = \mathop {\lim }\limits_{r \to \infty } {\mathcal L} (\bm W_{j,t}^{(r)}, \hat {\bm W}_t^{(r)}).
    \end{equation}
\end{lemma}
\begin{proof}
Please refer to Appendix A.
\end{proof}



\section{Online channel gain coefficient estimation}
Even though we obtain predicted locations of all ground users, it is still impossible to theoretically analyze (\ref{eq:original_problem}) without closed-form expressions of BtU and UtG channel gains.
To this end, we correlate channel gain values with channel gain coefficients in a simple manner (refer to (\ref{eq:UtG_channel_gain}) and (\ref{eq:BtU_channel_gain})). In this way, analytically tractable channel gain models can be achieved if the channel gain coefficients can be numerically generated.
In this paper, DNNs, which are of practical usage for complicated function approximations, are leveraged to generate numerical values of the channel gain coefficients.

\subsection{BtU channel gain coefficient estimation}
In this subsection, we start to describe the procedure of estimating BtU channel gain coefficient using a DNN from the viewpoint of the design of input and output space and network parameters. It is noteworthy that the BS can simultaneously connect to all UAVs and the channel gain coefficient of each BtU link should be estimated. Therefore, the BS is imposed to construct and train a DNN for each UAV.
For a typical UAV $j \in {\mathcal N}_s^u$, $s \in {\mathcal S}^u$, the input space, network parameters, and output spaces of the constructed DNN are presented as follows:
\begin{itemize}
\item \textbf{Input space $\bm s^{\rm ul}(t)$:} At each time slot $t$, we set $\bm s^{\rm ul}(t) = [{\bm v}_j^{\rm 3D}(t); {\bm x}_{\rm B}^{\rm 3D}; a_{j,s}^{\rm ul}(t) ]$. $a_{j,s}^{\rm ul}(t)$ is an indicator variable indicating whether there is an LoS link between the BS and UAV $j$ at time slot $t$. $a_{j,s}^{\rm ul}(t) = 1$ if an LoS link exists; otherwise, $a_{j,s}^{\rm ul}(t) = 0$. This is because the BtU channel gain is closely-related to the locations of the BS and UAV $j$ and the LoS/NLoS connection between them. In an area with given building locations and heights, the presence/absence of LoS link between the BS and a UAV can be exactly determined by checking whether the link connecting the BS and the UAV is blocked or not by any building.
\item \textbf{Network parameters $\theta^{\mu_j} (t)$:} We consider two fully-connected hidden layers. The \emph{ReLU} function is utilized as the activation function in the hidden layers. Besides, the network parameters are initialized by a Xavier initialization scheme \cite{glorot2010understanding}.
\item \textbf{Output space $o_j(t)$:} $o_j(t)$ is the estimated output value of the output layer. We set $\hat o_j(t) = \hat {\theta}_{j}^{\rm B}(t)$ with $\hat {\theta}_{j}^{\rm B}(t)$ being the target channel gain coefficient value, which can be obtained by signal measurements, between the BS and UAV $j$.
The acquisition of signal measurements is practically available based on the existing cellular communication protocols like the reference signal received power (RSRP) {for the signal receiving power measurement}, received signal strength indicator (RSSI) {for the total receiving power measurement}  \cite{zeng2020simultaneous}.
Besides, the \emph{linear} function is considered as the activation function in the output layer.
\end{itemize}

To effectively train the DNN, the experience replay technique \cite{mnih2015human} is exploited. This is due to the two special characteristics of the channel gain estimation issue: 1) the collected input data $\bm s^{\rm ul}(t)$ incrementally arrive as UAV $j$ moves to new locations, instead of all made available at the beginning of the training; 2) as UAV $j$ collects input data as it flies, those input data obtained at consecutive time slots are typically correlated, which may result in the oscillation or divergence of the DNN.
Specifically, at time slot $t$, a new training sample $(\bm s^{\rm ul}(t), \hat o_j(t))$ is added to a database or replay memory. When the memory is filled, the newly generated sample replaces the oldest one. We randomly choose a minibatch of training samples $\{(\bm s^{\rm ul}(\tau), \hat o_j(\tau))| \tau \in {\mathcal T}_t\}$ from the database, where ${\mathcal T}_t$ is a set of time slot indices. The network parameters $\theta^{u_j}(t)$ are trained using the ADAM method \cite{Kingma2014Adam} to reduce the averaged differences of squares, as
\begin{equation}\label{eq:mse}
L({\theta ^{{\mu _j}}(t)}) = \frac{1}{{|{{\mathcal T}_t}|}}\sum\nolimits_{\tau  \in {{\mathcal T}_t}} {{{\left( {{\hat {o}_j}(t) -  o_j(t)} \right)}^2}}.
\end{equation}

The following Algorithm \ref{alg_BtU_channel_gain} summarizes the steps of training the DNN for BtU channel gain estimation.
\begin{algorithm}[!htp]
\caption{DNN for BtU channel gain estimation}
\label{alg_BtU_channel_gain}
\begin{algorithmic}[1]
\STATE {\textbf {Initialize:}} DNN $\mu(\bm s^{\rm ul}(t)|\theta^{\mu_j} (t))$ with network parameters $\theta^{\mu_j} (t)$.
\STATE {\textbf {Initialize:}} Replay buffer $R$ with capacity $C$ and minibatch size $|{{\mathcal T}_t}|$
\FOR{each slot $t = 1, 2, \ldots$}
\STATE Calculate the input space $\bm s^{\rm ul}(t)$ according to the locations of the BS and UAV $j$.
\STATE The BS measures signals to obtain the target channel gain coefficient $\hat o_j(t)$.
\STATE Store the transition $(\bm s^{\rm ul}(t), \hat o_j(t))$ in the buffer $R$.
\STATE If $t \ge |{{\mathcal T}_t}|$, sample a random minibatch of $|{{\mathcal T}_t}|$ transitions $(\bm s^{\rm ul}(m), \hat o_j(m))$ from $R$.
\STATE Update the network parameters ${\theta ^{{\mu _j}}(t)}$ by minimizing the loss $L({\theta ^{{\mu _j}}(t)})$ using the ADAM method.
\ENDFOR
\end{algorithmic}
\end{algorithm}

In Algorithm \ref{alg_BtU_channel_gain}, the
training process at each time slot is fast via calling the ADAM method to reduce the value of $L({\theta ^{{\mu _j}}(t)})$ of a minibatch of randomly selected samples.
Thus, Algorithm \ref{alg_BtU_channel_gain} can be performed online. 

\subsection{UtG channel gain coefficient estimation}
Owing to the movement of UAVs and ground users, each UAV will have the possibility of connecting each ground user after a long enough period of time. Besides, the UtG channel gain is closely-related to the environment where UAVs are deployed and the locations of UAVs and users, regardless of exclusive characteristics (e.g., user profile information and device types) of ground users.
Therefore, each UAV is imposed to construct and train a unique DNN to estimate the UtG channel gain coefficients between it and all ground users in this paper.
Particularly, for a typical UAV $j$, we design the input space, network parameters, and output space of the DNN as follows:
\begin{itemize}
\item \textbf{Input space $\bm s^{{\rm dl}}(t)$:} At each time slot $t$, we set $\bm s^{{\rm dl}}(t) = [{\bm {\tilde x}}^{\rm 3D}(t); {\bm v}_j^{\rm 3D}(t); a_{j}^{\rm dl}(t) ]$. ${\bm {\tilde x}}^{\rm 3D}(t)$ is the 3D coordinate of the user connecting to UAV $j$. $a_{j}^{\rm dl}(t)$ is an indicator variable indicating whether there is an LoS link between UAV $j$ and its connected user at time slot $t$. $a_{j}^{\rm dl}(t) = 1$ if the LoS link exists; otherwise, $a_j^{\rm dl}(t) = 0$. The approach of determining $a_j^{\rm dl}(t)$ is similar to that of obtaining $a_{j,s}^{\rm ul}(t)$.
\item \textbf{Network parameters $\theta^{Q_j} (t)$:} We consider two fully-connected hidden layers.
The \emph{ReLU} is used as an activation function in hidden layers. The network parameters are initialized by a Xavier initialization scheme.
\item \textbf{Output space $o_{ij}^{{\rm dl}}(t)$:} $o_{ij}^{{\rm dl}}(t)$ is the output value of the output layer. We set the target output value $\hat o_{ij}^{{\rm dl}}(t) = \hat \theta_{ij}(t)$ with $\hat \theta_{ij}(t)$ being the target channel gain value between UAV $j$ and its connected user $i$ at $t$. Similarly, $\hat \theta_{ij}(t)$ can be obtained by signal measurements.
Besides, the \emph{linear} function is considered as the activation function in the output layer.
\end{itemize}

Likewise, the experience replay technique joint with the ADAM method is used to help train the DNN $Q(\bm s^{{\rm dl}}(t)|\theta^{Q_j}(t))$. The main steps of training the DNN for UtG channel gain estimation are similar to Algorithm \ref{alg_BtU_channel_gain}; yet, each UAV will measure signals for the DNN training. Thus, we omit the summarization of the main steps for brevity.

\section{Lyapunov-Based Optimization Framework}
\subsection{Lyapunov decomposition}
Given the predicted users' locations and estimated channel gain coefficients, it is still difficult to solve (\ref{eq:original_problem}). This is because (\ref{eq:original_problem}) is a sequential decision problem.
To tackle this challenge, we propose to decompose the sequential decision problem into multiple repeated optimization subproblems via the Lyapunov approach.
Particularly, let ${\bm \gamma} (t)=({{\gamma }_{1,1}}(t),\ldots ,{{\gamma }_{N_{|{\mathcal S}^e|}^e,|{\mathcal S}^e|}}(t))$ be an auxiliary data rate vector with {$0\le {{\gamma }_{i,s}}(t)\le u_{i,s}^{\rm max}(t)$, $\forall i \in {\mathcal N}_s^e, s\in {\mathcal S}^e, t$}. Define $g(t)=\phi (\bm{\gamma} (t)) = {\sum\nolimits_{i,s}{\log_2(1 + \gamma_{i,s}(t))}}$. The following lemma shows a transformed problem of (\ref{eq:original_problem}) and presents the conditions required to enforce the time average constraints of the transformed problem.
\begin{lemma}\label{lem:lemma_equivalent}
{\rm The original problem (\ref{eq:original_problem}) can be equivalently transformed into the following problem.
\begin{subequations}\label{eq:Jensen_problem}
\begin{alignat}{2}
& \mathop {\rm Maximize }\limits_{{\mathcal A}(t),{\mathcal P}(t),{{{\mathcal X}}}(t),{\bm {\gamma} (t)}} {\mkern 1mu} \text{ } \mathop {\lim \inf }\limits_{t \to \infty } \left( {\bar g(t) - \rho \sum\nolimits_{j \in {\mathcal J}} {{{{\bar p}_j^{\rm tot}}(t)}} } \right) \\
&{\rm s.t:} \quad {\mathop {\lim \inf }\limits_{t \to \infty } [{{\bar u}_{i,s}}(t) - {{\bar \gamma }_{i,s}}(t)] \ge 0,\forall i, s} \\
& \quad \mathop {\lim \inf }\limits_{t \to \infty } [{{\bar u}_{i,s}}(t) - C_s^{\rm th}] \ge 0,\forall i, s \allowdisplaybreaks[4] \\
& \quad \mathop {\lim \sup }\limits_{t \to \infty } [{{\tilde p}_j} - {{\bar p}_j^{\rm tot}}(t)] \ge 0,\forall j \\
& \quad 0 \le {\gamma _{i,s}}(t) \le u_{i,s}^{\rm max }(t), \forall i,s,t \\
& \quad {\rm (\ref{eq:original_problem}d)-(\ref{eq:original_problem}i).}
\end{alignat}
\end{subequations}

For all $i \in {\mathcal N}_s^e$, $s \in {\mathcal S}^e$, $j \in {\mathcal J}$, introduce three families of virtual queue variables $\{Q_{i,s}(t)\}$, $\{Z_{i,s}(t)\}$, $\{H_j(t)\}$, and update them with
\begin{equation}\label{eq:Queue_ui}
{Q_{i,s}}(t ) = {Q_{i,s}}(t-1) + C_s^{\rm th} - {u_{i,s}}(t-1),
\end{equation}
\begin{equation}\label{eq:Queue_Z}
{Z_{i,s}}(t ) = {Z_{i,s}}(t-1) + {\gamma _{i,s}}(t-1) - {u_{i,s}}(t-1),
\end{equation}
\begin{equation}\label{eq:Queue_H}
{H_j}(t ) = {H_j}(t-1) + p_j^{\rm tot}(t-1) - {\tilde p_j}.
\end{equation}
If the following mean-rate stability conditions can be held
\begin{equation}\label{eq:Queue_EQ}
\mathop {\lim }\nolimits_{t \to \infty } {{{\mathbb E}\{{{[{Q_{i,s}}(t)]}^ + }\}}}/{t} = 0,\text{ }\mathop {\lim }\nolimits_{t \to \infty } {{{\mathbb E}\{{{[{Z_{i,s}}(t)]^+} }\}}}/{t} = 0,\text{ }\mathop {\lim }\nolimits_{t \to \infty } {{{\mathbb E}\{{{[{H_j}(t)]}^ + }\}}}/{t} = 0
\end{equation}
where the non-negative operation ${[x]^ + } = \max \{ x,0\}$, then the time average constraints of (\ref{eq:Jensen_problem}) can be satisfied.}
\end{lemma}
\begin{proof}
Please refer to Appendix B. 
\end{proof}

Then the following question pops up: \emph{how to solve (\ref{eq:Jensen_problem}) with the virtual queues?}

For simplicity, we assume that all virtual queues are initialized to be zero and define a Lyapunov function $L\left( t \right)$ as a sum of square of all the three virtual queues ${{[{{Q}_{i,s}}(t)]}^{+}}$, ${[{{{Z}_{i,s}}(t)}]^+}$ and ${{[{{H}_{j}}(t)]}^{+}}$ (divided by 2 for convenience) at $t$, i.e.,
${L(t)} \text{ } { \buildrel \Delta \over = \text{ } \frac{1}{2}\sum\nolimits_{i,s} {{{({{[{Q_{i,s}}(t)]}^ + })}^2}}  + \frac{1}{2}\sum\nolimits_{i,s} {{{({{[{Z_{i,s}}(t)]^+} })}^2}}  }$ $ + {\frac{1}{2}\sum\nolimits_{j \in {\mathcal J}} {{{({{[{H_j}(t)]}^ + })}^2}} }$.
$L(t)$ is a scalar measure of constraint violations. Intuitively, if the value of $L(t)$ is small, the absolute values of all queues are small; otherwise, the absolute value of at least one queue is great.
Additionally, we define a drift-plus-penalty function as $\Delta (t)-V\left( g(t)-\rho \sum\nolimits_{j \in {\mathcal J}}{ {{p}_{j}^{\rm tot}}(t)} \right)$, where $\Delta (t) = L(t + 1) - L(t)$ represents a Lyapunov drift, $-V\left( g(t)-\rho \sum\nolimits_{j \in {\mathcal J}}{{{p}_{j}^{\rm tot}}(t)} \right)$ is a \emph{penalty}, and $V$ is a non-negative penalty coefficient that weighs a trade-off between the constraint violations and the optimality. In this way, the solution of sequential decision problem (\ref{eq:Jensen_problem}) can be implemented by \emph{repeatedly minimizing the drift-plus-penaly function under all non-time average constraints of (\ref{eq:Jensen_problem}) at each time slot $t$.} 

Besides, the following lemma presents the upper bound of the drift-plus-penaly function value.
\begin{lemma}\label{lemma:1}
{\rm At each time slot $t$, the upper bound of the value of the drift-plus-penalty function $\Delta (t)-V( g(t)-$ $\rho \sum\nolimits_{j \in {\mathcal J}}{{{p}_{j}^{\rm tot}}(t)} )$ can be expressed as (\ref{eq:upper_bound}) with $B{\rm{ }} \buildrel \Delta \over = {\sum\nolimits_{i,s} {{{(u_{i,s}^{\max })}^2}}}  + \sum\nolimits_{j \in {\mathcal J}} {{{(\hat p_j)}^2}/2}$.
\begin{equation}\label{eq:upper_bound}
\begin{array}{l}
\Delta (t) - V\left( {g(t) - \rho \sum\nolimits_{j \in {\cal J}} {p_j^{\rm tot}(t)} } \right) \le B +
\sum\nolimits_{i,s} {{{[{Q_{i,s}}(t)]}^ + }C_s^{\rm th}} - \sum\nolimits_{j \in {\cal J}} {{{[{H_j}(t)]}^ + }\left( {{{\tilde p}_j} - p_j^c} \right)}  +
V\rho \sum\nolimits_{j \in {\cal J}} p_j^c - \\
\quad V\phi ({\bm \gamma} (t)) + \sum\nolimits_{i,s} {{{[{Z_{i,s}}(t)]^+} }{\gamma _{i,s}}(t)}   +
\sum\nolimits_{j \in {\cal J}} {\{ V\rho  + {{[{H_j}(t)]}^ + }\} {p_j}(t)}  -
\sum\nolimits_{i,s} {\{ {{[{Q_{i,s}}(t)]}^ + } + {{[{Z_{i,s}}(t)]^+} }\} {u_{i,s}}(t)}.
\end{array}
\end{equation}
}
\end{lemma}

\begin{proof}
Please refer to Appendix C. 
\end{proof}

In (\ref{eq:upper_bound}), the right-hand-side expression is the upper bound of the drift-plus-penalty. As such, the minimization of the drift-plus-penalty can be approximated by minimizing its upper bound.
Further, the upper bound can be partitioned into two independent terms related to $\bm \gamma(t)$ and other sets of decision variables, i.e., ${\mathcal A}(t)$, ${\mathcal P}(t)$, ${\mathcal X}(t)$, respectively.
Therefore, we can summarize the Lyapunov-based optimization framework of mitigating (\ref{eq:original_problem}) as follows.
\begin{itemize}
\item At each $t$, observe $Q_{i,s}(t)$, ${Z_{i,s}}(t)$, ${{H}_{j}}(t)$ for all $i \in {\mathcal N}_s^e$, $s \in {\mathcal S}^e$, $j \in {\mathcal J}$.
\item Choose ${{\gamma _{i,s}}(t)}$ for each user $i$ to mitigate (\ref{eq:gamma_related_problem})
\begin{subequations}\label{eq:gamma_related_problem}
\begin{alignat}{2}
& \mathop {\rm Minimize }\limits_{{\bm {\gamma}} (t)} {\mkern 1mu} \text{ } \mathop  - V\phi ({\bm {\gamma} (t)}) + {\sum\nolimits_{i,s}{{{[{Z_{i,s}}(t)]^+} }{\gamma _{i,s}}(t)}} \\ &{\rm {s.t:} \text{ }} \quad 0 \le \gamma_{i,s}(t) \le u_{i,s}^{\max}(t).
\end{alignat}
\end{subequations}
\item Given UAVs' locations ${{\mathcal {X}}(t-1)}$, choose ${{\mathcal {A}}(t)}$, ${{\mathcal {P}}(t)}$, and ${{ {{\mathcal X}}}(t)}$ to mitigate (\ref{eq:subproblem_BPX})
\begin{subequations}\label{eq:subproblem_BPX}
\begin{alignat}{2}
& \mathop {\rm Minimize }\limits_{{\mathcal {A}}(t),{\mathcal {P}}(t),{{{\mathcal X}}(t)}} {\mkern 1mu} \text{ }  \sum\nolimits_{j \in {\mathcal J}} {\{ V\rho  + {{[{H_j}(t)]}^ + }\} {p_j}(t)}  - {\sum\nolimits_{i,s}{\{ {{[{Q_{i,s}}(t)]}^ + } + {{[{Z_{i,s}}(t)]^+} }\} {u_{i,s}}(t)}}   \\
& {\rm s.t:} \quad {\rm (\ref{eq:original_problem}d)-(\ref{eq:original_problem}i)}.
\end{alignat}
\end{subequations}
\item Compute $u_{i,s}(t)$ using (\ref{eq:achievable_data_rate}). Update three virtual queues using (\ref{eq:Queue_ui}), (\ref{eq:Queue_Z}), and (\ref{eq:Queue_H}).
\end{itemize}

\subsection{Solution to subproblem (\ref{eq:gamma_related_problem})}
The proportional fairness function $\phi (\bm \gamma (t))$ is a separable sum of individual logarithmic functions. Thus, the mitigation of (\ref{eq:gamma_related_problem}) is equivalent to a separate selection of the individual auxiliary variable {${{\gamma }_{i,s}}(t)\in [ 0,{u_{i,s}^{\max }(t)} ]$} for each user $i \in {\mathcal N}_s^e$, $s\in{\mathcal S}^e$ that minimizes the convex function $-V{{\log }_{2}}(1+{{\gamma }_{i,s}}(t))+{{[{{Z}_{i,s}}(t)]}^{+}}{{\gamma }_{i,s}}(t)$. In consequence, the closed-form solution to (\ref{eq:gamma_related_problem}) can be written as
\begin{equation}\label{eq:compute_gamma}
{\gamma _{i,s}}(t) = \left\{ {\begin{array}{*{20}{l}}
{u_{i,s}^{\max }(t), \quad {{{[{Z_{i,s}}(t)]^+} } = 0},}\\
{\min \left\{ {{{\left[ {\frac{V}{{{{[{Z_{i,s}}(t)]^+} }\ln 2}} - 1} \right]}^ + },\;u_{i,s}^{\max }(t)} \right\},\text{ } {\rm otherwise.}}
\end{array}} \right.
\end{equation}



\subsection{Solution to subproblem (\ref{eq:subproblem_BPX})}
To make (\ref{eq:subproblem_BPX}) easier to tackle, we attempt to reduce its variable dimension and decompose it into several independent subproblems. The key observations on (\ref{eq:subproblem_BPX}) are as follows: in (\ref{eq:subproblem_BPX}), network resources allocated to URLLC slices and MBB slices should be simultaneously optimized; the system bandwidth allocated to URLLC slices and MBB slices, however, is only coupled at the total bandwidth constraint {(\ref{total_bandwidth})}. We therefore decompose (\ref{eq:subproblem_BPX}) into a subproblem of URLLC slice resource allocation and a subproblem of MBB slice resource allocation via decoupling {(\ref{total_bandwidth})}.

\subsubsection{Subproblem formulation of URLLC slice resource allocation}
In (\ref{con_urllc_power}), the URLLC bandwidth $W^u(t)$ and the total transmit power $p(W^u(t))$ are correlated. By referring to (\ref{con_urllc_power}), the following lemma, which reflects the relationship between $p(W^u(t))$ and $W^u(t)$, can be derived.
\begin{lemma}\label{lem:power_bandwidth}
$p(W^u(t))$ satisfies the following properties \cite{sun2018optimizing}:
\begin{itemize}
\item $p(W^u(t))$ decreases with $W^u(t)$ when $0 < W^u(t) \le {W_u^{{\rm th}}}$, where $W_u^{{\rm th}}$ is the unique solution that minimizes $p(W^u(t))$.
\item The derivative of $p(W^u(t))$, denoted by $\frac{{\partial p({W^u}(t))}}{{\partial {W^u}(t)}}$, satisfies
\begin{equation}\label{der_y}
\frac{{\partial p({W^u}(t))}}{{\partial {W^u}(t)}}\left\{ {\begin{array}{*{20}{l}}
{ < 0,}&{0 < {W^u}(t) \le W_u^{{\rm{th}}}}\\
{ = 0,}&{{W^u}(t) = W_u^{{\rm{th}}}}\\
{ > 0,}&{{W^u}(t) > W_u^{{\rm{th}}}}
\end{array}} \right.
\end{equation}
\end{itemize}
\end{lemma}
This lemma shows that there is a value $W_u^{\rm th}$ to minimize the total required BS transmit power.
In consequence, minimizing $W^u(t)$ while satisfying the BS transmit power constraint will not decrease the objective function value of (\ref{eq:subproblem_BPX}).
On the contrary,
finding out the minimum $W^u(t)$ satisfying the BS transmit power constraint can lead to the increase of the value of $W^e(t)$. As $u_{i,s}(t)$ monotonously increases with $W^e(t)$ \cite{sun2018optimizing}, the maximum value of the objective function of (\ref{eq:subproblem_BPX}) can be obtained when the minimum $W^u(t)$ is achieved.
It indicates that the optimal solution can still be obtained even if the subproblem decomposition scheme is performed.

Therefore, we can formulate the subproblem of URLLC slice resource allocation at slot $t$ as follows:
\begin{subequations}\label{eq:URLLC_subproblem}
\begin{alignat}{2}
& \mathop {\min }\limits_{{W^u(t)},\{ p_{j,s}^{\rm B}(t) \} } \quad {{W^u}(t)} \allowdisplaybreaks[4]  \\
& {\rm s.t:} \quad 0 < {W^u}(t) \le W^{\rm tot}, \forall t \\
& \quad {\rm (\ref{con_urllc_power}) \text{ } is \text{ } satisfied}.
\end{alignat}
\end{subequations}

Based on the second property in Lemma \ref{lem:power_bandwidth}, we can conclude that (\ref{eq:URLLC_subproblem}) is non-convex, which is challenging to be solved.
To solve the problem efficiently, a binary search method \cite{sun2018optimizing} was executed on the condition $\frac{{\partial p(W^u(t))}}{{\partial {{\bar W}^u(t)}}} = 0$ to obtain ${W_u^{\rm th}(t)}$.
Besides, the mitigation of (\ref{eq:URLLC_subproblem}) is equivalent to obtaining $W^u(t)$ satisfying $p(W^u(t)) = {p_{\rm B}^{\rm max}}$ constrained on ${W^u(t)} \in (0,W_{ub}^u(t)]$ with $W_{ub}^u(t) = \min \left\{ {{W_u^{\rm th}(t)},W^{\rm tot}} \right\}$. Similarly, the binary search method is leveraged to obtain the minimum URLLC bandwidth $W_{opt}^u(t)$ via searching $W^u(t)$ in the interval $(0,W_{ub}^u(t)]$.

As $p(W_{opt}^u(t)) = p_{\rm B}^{\max}$, we have
\begin{equation}\label{eq:opt_URLLC_bandwidth}
p_{j,s}^{\rm B}(t) = {{p_{\rm B}^{\rm max}} \mathord{\left/
 {\vphantom {{p_{\rm B}^{\rm max}} {\left( {{{\left| {h_{j,s}^{\rm B}(t)} \right|}^2}\sum\nolimits_{l \in {S_{uav}}} {\frac{1}{{{{\left| {h_l^u} \right|}^2}}}} } \right)}}} \right.
 \kern-\nulldelimiterspace} {\left( {{{\left| {h_{j,s}^{\rm B}(t)} \right|}^2}\sum\nolimits_{s \in {\mathcal S}^u}\sum\nolimits_{l \in {\mathcal N}_s^u} {{{{{\left| {h_{l,s}^{\rm B}(t)} \right|}^{-2}}}}} } \right)}}.
\end{equation}

\subsubsection{Subproblem of MBB slice resource allocation}
Given $W^u(t)$, the subproblem of MBB slice resource allocation can be formulated as follows:
\begin{subequations}\label{eq:MBB_subproblem}
\begin{alignat}{2}
& \mathop {\rm Minimize }\limits_{{\mathcal {A}}(t),{\mathcal {P}}(t),{{{\mathcal X}}(t)}} {\mkern 1mu} \text{ }  \sum\nolimits_{j \in {\mathcal J}} {\{ V\rho  + {{[{H_j}(t)]}^ + }\} {p_j}(t)}  - \nonumber \\
& \qquad \qquad {\sum\nolimits_{i,s}{\{ {{[{Q_{i,s}}(t)]}^ + } + {{[{Z_{i,s}}(t)]^+} }\} {u_{i,s}}(t)}}   \\
& {\rm s.t:} \quad {\rm (\ref{eq:original_problem}d)-(\ref{eq:original_problem}g),(\ref{eq:slice_request_indicator}),(\ref{total_bandwidth})}.
\end{alignat}
\end{subequations}

In (\ref{eq:MBB_subproblem}), there are continuous variable sets ${\mathcal P}(t)$, ${\mathcal X}(t)$ and 0-1 variable set ${\mathcal A}(t)$. Besides, the constraint (\ref{eq:original_problem}f) is non-convex over $\bm v_j(t)$, $\forall j$. (\ref{eq:MBB_subproblem}a) is non-convex over $\bm v_j(t)$ and $p_j(t)$, $\forall j$. Therefore, (\ref{eq:MBB_subproblem}) is a challenging mixed-integer-non-convex programming problem.

The alternative optimization scheme has been shown to be an effective scheme of solving mixed-integer programming problems \cite{abbasi2020trajectory}. We therefore adopt this type of scheme to mitigate (\ref{eq:MBB_subproblem}) and first attempt to optimize the acceptance of MBB slice requests. Supported by the optimal slice access enforcement, UAV movement control and UAV transmit power optimization are then performed.

\emph{a) Acceptance optimization of slice requests:} Given UAVs' locations ${\mathcal X}(t)$ and UAV transmit power ${\mathcal P}(t)$, the acceptance set of slice requests of (\ref{eq:MBB_subproblem}) can be optimized by mitigating the following problem
\begin{equation}\label{eq:UE_association_problem}
\mathop {\rm Maximize }\limits_{{\mathcal {A}}(t)} {\mkern 1mu} \text{ } {\sum\nolimits_{i,s} {\sum\nolimits_{j \in {\mathcal J}} {{c_{ij,s}(t)}{a_{ij,s}}(t)} } }  \quad {\rm s.t: \text{ }} {\rm (\ref{eq:slice_request_indicator}),(\ref{eq:original_problem}g).}
\end{equation}
where ${c_{ij,s}(t)} = \{ {[{Q_{i,s}}(t)]^ + } + {[{Z_{i,s}}(t)]^ + }\}W^e(t){\log _2}\left( {1 + {\rm SINR}_{ij,s}(t)} \right)$.

Note that at the initial time slot $t = 1$, all weights $\{c_{ij,s}(1)\}$ equal to zero since all virtual queues are initialized to be zero. To tackle this issue, we define the weight $c_{ij,s}(1)$ as $c_{ij,s}(1) = W^e(t)\log_2 ( {1 + {\rm SINR}_{ij,s}(1)} )
$. (\ref{eq:UE_association_problem}) is a linear integer programming problem and can be efficiently solved by some optimization tools such as MOSEK.

\emph{b) UAV movement control:} As the constraint (\ref{eq:original_problem}f) is non-convex over $\bm v_j(t)$, $\forall j,t$, the SCA technique \cite{scutari2017parallel} is exploited to tackle the non-convexity.
The key idea of SCA is to solve a sequence of convex optimization problems with different initial points to obtain an approximate solution to a non-convex optimization problem instead of solving the non-convex problem directly. In this paper, we first utilize some approximation functions to locally approximate non-convex (\ref{eq:MBB_subproblem}a) and (\ref{eq:original_problem}f) based on the following assumption.
\begin{assumption}\label{assumption_1}
A function $\tilde f:{\mathcal X} \times {\mathcal K} \to {\mathbb R}$ is called the approximation
function for the non-convex function $f(x)$ ($x \in {\mathcal X}$), when the following conditions hold \cite{scutari2017parallel}:
\begin{itemize}
\item $\tilde f(\cdot, \cdot)$ is continuous in ${\mathcal X} \times {\mathcal K}$.
\item $\tilde f(\cdot, x^{(r)})$ is convex in ${\mathcal X}$ for all $x^{(r)} \in {\mathcal K}$.
\item Function value consistency: $\tilde f(x^{(r)}, x^{(r)}) = f(x^{(r)})$ for all $x^{(r)} \in {\mathcal X}$.
\item Gradient consistency: $\frac{{\partial \tilde f(x,{x^{(r)}})}}{{\partial x}}{|_{x = {x^{(r)}}}} = \nabla f(x){|_{x = {x^{(r)}}}}$ for all $x^{(r)} \in {\mathcal K}$.
\item Upper bound: $f(x) \le \tilde f(x, x^{(r)})$ for all $x \in {\mathcal X}$, $x^{(r)} \in {\mathcal K}$.
\item $\frac{{\partial \tilde f(\cdot, \cdot)}}{{\partial x}}$ is continuous in ${\mathcal X} \times {\mathcal K}$.
\end{itemize}
\end{assumption}

Then, given the UAV transmit power ${\mathcal P}(t)$, UAVs' locations ${{\mathcal {X}}(t-1)}$ at the previous time slot $t-1$, and the acceptance set of slice requests ${\mathcal A}(t)$, the following proposition shows a method of controlling the movement of UAVs.
\begin{proposition}\label{lemma:lemma_uav_location}
{\rm By exploring the SCA technique, UAVs' locations at $t$ can be obtained by mitigating the following convex optimization problem:
\begin{subequations}\label{eq:UAV_location_equal_problem_approximate}
\begin{alignat}{2}
& \mathop {\rm Maximize }\limits_{{\mathcal {X}}(t), \{\eta_{i,s}(t)\}, \{B_{ik,s}(t)\}} {\mkern 1mu} \text{ } {{\sum\nolimits_{i,s} {\{ {{[{Q_{i,s}}(t)]}^ + } + {{[{Z_{i,s}}(t)]^+} }\} {\eta _{i,s}}(t)} }}  \\
& {\rm s.t:} \sum\nolimits_{j \in {\cal J}} {a_{ij,s}}(t){W^e}(t)\left( D_{i,s}^{(r)}(t) - \sum\nolimits_{k \in {\cal J}} E_{ik,s}^{(r)}(t)\left( ||{{\bm v}_k}(t) - {{\bm x}_{i,s}}(t)|{|^2} - ||{\bm v}_k^{(r)}(t) - {{\bm x}_{i,s}}(t)|{|^2} \right) \right) + \nonumber \\
& \qquad \qquad \sum\nolimits_{j \in {\cal J}} {{a_{ij,s}}(t){W^e}(t){{\tilde R}_{ij,s}}(t)}  \ge {\eta _{i,s}}(t),\forall i,t \allowdisplaybreaks[4] \\
& \quad {B_{ik,s}(t)} \le  {|| {{\bm v}_k^{(r)}(t) - {\bm x}_{i,s}(t)} ||^2} + 2{( {{\bm v}_k^{(r)}(t) - {\bm x}_{i,s}(t)} )^{\rm T}}\left( {{{\bm v}_k}(t) - {{\bm x}_{i,s}(t)}} \right),\forall i,k \ne j,t \\
& \quad - {|| {{\bm v}_j^{(r)}(t) - {\bm v}_k^{(r)}(t)} ||^2} + 2{( {{\bm v}_j^{(r)}(t) - {\bm v}_k^{(r)}(t)} )^{\rm T}} \left( {{{\bm v}_j}(t) - {{\bm v}_k}(t)} \right) \ge d_{\rm min}^2, \forall j,k \ne j,t \\
& \quad {\rm (\ref{eq:original_problem}e) \text{ } is \text{ } satisfied,}
\end{alignat}
\end{subequations}
where $\eta_{i,s}(t)$ and ${B_{ik,s}(t)}$ are slack variables, $D_{i,s}^{(r)}(t) = {\log _2}( {{N_0W^e(t)} + \sum\limits_{k \in {\mathcal J}} {\frac{{{{p_k}(t)\theta _{ik,s}(t)}}}{{{|g_k(t)-g_{i,s}|^2} + ||{\bm v}_k^{(r)}(t) - {{\bm x}_{i,s}(t)}|{|_2^2}}}} } )$, $E_{ik,s}^{(r)}(t) = {{\frac{{{{p_k}(t)\theta _{ik,s}(t)}}}{{{{\left( {{|g_k(t)-g_{i,s}|^2} + ||{\bm v}_k^{(r)}(t) - {{\bm x}_{i,s}(t)}|{|_2^2}} \right)}^2}{2^{D_{i,s}^{(r)}(t)}\ln2}}}}}$, ${\tilde R_{ij,s}}(t) =  - {\log _2}( {{N_0W^e(t)} + \sum\nolimits_{k \in {\mathcal J}\backslash \{ j\} } {\frac{{{{p_k}(t)\theta _{ik,s}(t)}}}{{{|g_k(t)-g_{i,s}|^2} + {B_{ik,s}}(t)}}} } )$, and ${\bm v}_{j}^{(r)}(t)$, ${\bm v}_{k}^{(r)}(t)$ are given locations at the $r$-th iteration.}
\end{proposition}
\begin{proof}
Please refer to Appendix D.
\end{proof}

{\textbf{Remark}}: As (\ref{eq:UAV_location_equal_problem_approximate}) is convex, the optimization tool MOSEK can be utilized to effectively solve it. Owing to the approximation, the feasible domain of (\ref{eq:UAV_location_equal_problem_approximate}) is smaller than that of (\ref{eq:MBB_subproblem}) with fixed ${\mathcal P}(t)$ and ${\mathcal A}(t)$; thus, the value of (\ref{eq:UAV_location_equal_problem_approximate}a) is the upper bound of the opposite of (\ref{eq:MBB_subproblem}a) with given ${\mathcal P}(t)$ and ${\mathcal A}(t)$, if it exists.

\emph{c) UAV transmit power optimization:} For any given slice request acceptance set ${\mathcal A}(t)$, UAVs' locations ${\mathcal X}(t)$, the following proposition shows a method of optimizing the UAV transmit power.
\begin{proposition}\label{lemma:lemma_UAV_power}
{\rm By exploiting the SCA technique, the UAV transmit power at $t$ can be optimized by mitigating the following convex optimization problem:
\begin{subequations}\label{eq:UAV_power_problem_approx}
\begin{alignat}{2}
& \mathop {\rm Maximize }\limits_{{\mathcal {P}}(t),{\{ \eta_{i,s}(t) \}}} {\mkern 1mu} \text{ } - V\rho \sum\nolimits_{j \in {\mathcal J}} {{p_j}(t)}  - \sum\nolimits_{j \in {\mathcal J}} {{{[{H_j}(t)]}^ + }{p_j}(t)}  + {\sum\nolimits_{i,s} {\{ {{[{Q_{i,s}}(t)]}^ + } + {{[{Z_{i,s}}(t)]^+} }\} {\eta _{i,s}}(t)}} \\
& {\rm s.t:} \sum\limits_{j \in {\mathcal J}} {( {{a_{ij,s}}(t)W^e(t)({{\hat R}_{i,s}}(t) - F_{ij,s}^{(r)}(t))} )}  - \sum\limits_{j \in {\mathcal J}} \left ( {a_{ij,s}}(t)W^e(t) \right.  \sum\limits_{k \in {\mathcal J}\backslash \{ j\} } {G_{ik,}^{(r)}(t)( {{p_k}(t) - p_k^{(r)}(t)} )}  ) \ge {\eta _{i,s}}(t),\forall i,s,t  \allowdisplaybreaks[4]  \\
& {\rm (\ref{eq:original_problem}d) \text{ } is \text{ } satisfied,}
\end{alignat}
\end{subequations}
where ${\hat R_{i,s}}(t) = {\log _2}( {{N_0W^e(t)} + \sum\limits_{k \in {\mathcal J}} {\frac{{{{p_k}(t)\theta _{ik,s}(t)}}}{{{|g_k(t)-g_{i,s}|^2} + ||{{\bm v}_k}(t) - {{\bm x}_{i,s}}|{|^2}}}} } )$, $F_{ij,s}^{(r)}(t) = {\log _2}( {N_0W^e(t)} + \sum\limits_{k \in {\mathcal J}\backslash \{ j\} } {p_k^{(r)}(t){h_{ik,s}}(t)}  )$, $G_{ik,s}^{(r)}(t) = \frac{{{h_{ik,s}}(t)}}{2^{F_{ij,s}^{(r)}(t)}\ln 2}$, and $p_k^{(r)}(t)$ is the given transmit power of UAV $k$ at the $r$-th iteration.}
\end{proposition}
\begin{proof}
Please refer to Appendix E.
\end{proof}

{\textbf{Remark}}: MOSEK tool can now be leveraged to effectively solve the convex (\ref{eq:UAV_power_problem_approx}). Likewise, the utilization of SCA technique results in the feasible domain of (\ref{eq:UAV_power_problem_approx}) being smaller than that of (\ref{eq:MBB_subproblem}) with given ${\mathcal A}(t)$ and ${\mathcal X}(t)$. Therefore, the optimal value of (\ref{eq:UAV_power_problem_approx}a) is the upper bound of the opposite of (\ref{eq:MBB_subproblem}a) with fixed ${\mathcal A}(t)$ and ${\mathcal X}(t)$, if it exists.

Based on the above derivation, we next propose an iterative algorithm, named iterative request, location and power optimization, for (\ref{eq:MBB_subproblem}) that is summarized as below.
\begin{algorithm}
\caption{Iterative request, location and power optimization}
\label{alg:alg1}
\begin{algorithmic}[1]
\STATE \textbf{Initialization:} Initialize ${\mathcal X}^{(0)}(t) = {\mathcal X}(t-1)$, ${\mathcal P}^{(0)}(t) = {\mathcal P}(t-1)$, and {$\bar r_{\rm max}$}, let $r = 0$.
\REPEAT
\STATE Given ${{\mathcal X}^{(r)}(t), {\mathcal P}^{(r)}(t)}$, solve (\ref{eq:UE_association_problem}) to obtain the optimal solution ${{\mathcal A}^{(r+1)}(t)}$.
\STATE Given ${{\mathcal A}^{(r+1)}(t), {\mathcal X}^{(r)}(t), {\mathcal P}^{(r)}(t)}$, solve (\ref{eq:UAV_location_equal_problem_approximate}) to generate the optimal solution ${{\mathcal X}^{(r+1)}(t)}$.
\STATE Given ${{\mathcal A}^{(r+1)}(t), {\mathcal X}^{(r+1)}(t), {\mathcal P}^{(r)}(t)}$, solve (\ref{eq:UAV_power_problem_approx}) to obtain the optimal solution ${{\mathcal P}^{(r+1)}(t)}$. Update $r = r + 1$.
\UNTIL {Convergence or $r = \bar r_{\rm max}$.}
\end{algorithmic}
\end{algorithm}

Finally, we can summarize the energy-efficient and fair algorithm of mitigating the UAV network slicing problem in Algorithm \ref{alg:final_algorithm}.
\begin{algorithm}
\caption{Repeatedly Energy-Efficient and Fair Service coverage, RE$^2$FS}
\label{alg:final_algorithm}
\begin{algorithmic}[1]
\STATE \textbf{Initialization:} Randomly initialize UAVs' locations, and run initialization steps of Algorithms \ref{alg_user_loc_prediction} and \ref{alg_BtU_channel_gain}.
\STATE \textbf{Initialization:} Let $Q_{i,s}(1) = 0$, $Z_{i,s}(1) = 0$, $H_j(1) = 0$ for all $i \in {\mathcal N}_s^e$, $s \in {\mathcal S}^e$, $j \in {\mathcal J}$.
\STATE {\textbf{Pre-train ESN models and DNNs:}}
\FOR {each episode $\hat t = 1, 2, \ldots, 500$}
\STATE Steps 3-19 of Algorithm \ref{alg_user_loc_prediction}.
\ENDFOR
\FOR {each episode $\hat t = 1, 2, \ldots, 3000$}
\STATE Steps 4-8 of Algorithm \ref{alg_BtU_channel_gain}. Pre-train the DNN for UtG channel gain coefficient estimation for 1000 episodes.
\ENDFOR
\FOR {each time slot $t = 1, 2, \ldots, T$}
\STATE \textbf{Predict users' locations using the distributed ESN learning method:}
\STATE Step 19 of Algorithm \ref{alg_user_loc_prediction} to obtain the predicted locations $\hat {\bm y}_{i,s}^{\rm B}(t+K)$ of user $i \in {\mathcal N}_s^e$ for all $s \in {\mathcal S}^e$.
\STATE \textbf{Estimate BtU and UtG channel gain coefficients using the DNNs:}
\STATE Observe the state $\bm s^{\rm ul}(t+K)$ and $\bm s^{\rm dl}(t+K)$. Use DNNs to estimate the corresponding channel gain coefficients $\theta_{ij,s}(t+K)$ and $\theta_j^{\rm B}(t+K)$. Then, calculate channel gains $h_{ij,s}(t+K)$, $\forall i \in {\mathcal N}_s^e$, $j\in {\mathcal J}$, $s\in {\mathcal S}^e$ and $h_{j,s}^{\rm B}(t+K)$, $\forall j \in {\mathcal N}_s^u$, $s \in {\mathcal S}^u$ using (\ref{eq:UtG_channel_gain}) and (\ref{eq:BtU_channel_gain}), respectively.
\STATE \textbf{Slice resource allocation:}
\STATE Compute $\gamma_{i,s}(t+K)$ using (\ref{eq:compute_gamma}) for all $i$ and $s$.
\STATE Call the binary search method to obtain the optimal $W^u(t+K)$, and calculate $W^e(t+K)$ using (\ref{total_bandwidth}).
\STATE Find the acceptance set of slice requests ${\mathcal A}(t+K)$, UAVs' locations ${\mathcal X}(t+K)$, and UAV transmit power ${\mathcal P}(t+K)$ using Algorithm \ref{alg:alg1}.
\STATE \textbf{Update the ESN models and DNNs:}
\STATE Steps 5-19 of Algorithm \ref{alg_user_loc_prediction}.
\STATE Steps 4-8 of Algorithm \ref{alg_BtU_channel_gain}. Likewise, train the DNN for UtG channel gain coefficient estimation.
\STATE \textbf{Update virtual queues:}
\STATE Calculate $u_{i,s}(t+K)$ for all $i$ and $s$ using (\ref{eq:achievable_data_rate}). Calculate $p_j^{\rm tot}(t+K)$ for all $j$ using (\ref{eq:power_total_at_t}).
\STATE Update $Q_{i,s}(t+K+1)$, $Z_{i,s}(t+K+1)$, and $H_j(t+K+1)$ for all $i$, $s$ and $j$ using (\ref{eq:Queue_ui}), (\ref{eq:Queue_Z}), and (\ref{eq:Queue_H}).
\ENDFOR
\end{algorithmic}
\end{algorithm}

\section{Implementation and performance analysis of Algorithm \ref{alg:final_algorithm}}
In this section, we first summarize the implementation of Algorithm \ref{alg:final_algorithm}. Then, we analyze the convergence and computational complexity of Algorithm \ref{alg:final_algorithm}.

In Algorithm \ref{alg:final_algorithm}, to ensure that accurate users' predicted locations and estimated channel gain coefficients can be inputted when calling steps 11-18, we perform steps 3-9 before executing the communication task.
Fig. \ref{fig_algorithm_flow} depicts the logical flow of Algorithm \ref{alg:final_algorithm}, where \textcircled{1} is firstly called for ESN model and DNN pre-training and then the logical flow \textcircled{2} $ \to $ \textcircled{3} $\to $ \textcircled{4} $\to$ \textcircled{5} $\to$ \textcircled{6} is executed.
Besides, Algorithm \ref{alg:final_algorithm} can effectively tackle the mismatch issue of slice supply and demand.
Although the process of virtually isolating the UAV network resources and functions is time-consuming, this process is desired to be completed within the time interval $(t, t+K)$ based on the predicted users' locations $\{\hat {\bm y}_{i,s}^{\rm B}(t+K)\}$, $\forall i,s,t$. At time slot $t+K$, the well-created and configured network slices will be utilized to accommodate the QoS requirements of UAV control and non-payload links and to serve ground mobile users.
Summarily, Algorithm \ref{alg:final_algorithm} allows to partition network slices in advance; thus, we call it \emph{proactive UAV network slicing}.
\begin{figure}[!t]
\centering
\includegraphics[width=5 in]{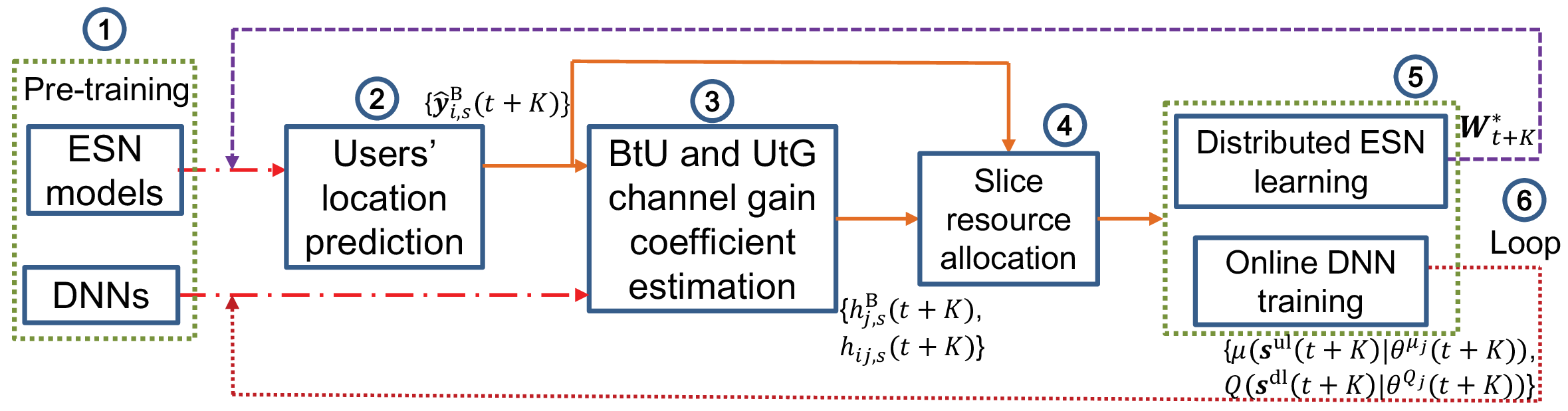}
\caption{Logical flow of Algorithm \ref{alg:final_algorithm}.}
\label{fig_algorithm_flow}
\end{figure}

The following lemma shows that the performance of Algorithms \ref{alg:alg1} and \ref{alg:final_algorithm} can be guaranteed.
\begin{lemma}\label{lemma:lemma_convergent}
{\rm Algorithm \ref{alg:alg1} is convergent, and Algorithm \ref{alg:final_algorithm} can make all virtual queues mean-rate stabilize.}
\end{lemma}
\begin{proof}
Please refer to Appendix F.
\end{proof}

We focus on the analysis of the computational complexity of Algorithm \ref{alg:final_algorithm} from the complexity perspective of two main contributors: Algorithm \ref{alg_user_loc_prediction} and the procedures of solving both (\ref{eq:gamma_related_problem}) and (\ref{eq:subproblem_BPX}).
For Algorithm \ref{alg_user_loc_prediction}, (\ref{eq:evolution_wj}), (\ref{eq:evolution_hatw}), (\ref{eq:lagrangian_aj}) must be iteratively computed until converge or reach the maximum iteration times $r_{\rm max}$. Besides, (\ref{eq:output}) should be called for $K$ times to obtain users' predicted locations. Therefore, the worst-case complexity of calling Algorithm \ref{alg_user_loc_prediction} is $O(f_1)=O(N^e(r_{\rm max}((N_i+N_r)^3N_o+(N_i+N_r)^2Q+(N_i+N_r)N_oQ)+(N_i+N_r)N_oK))$.
For the second contributor, it can be further divided into five parts, i.e., the complexities of solving (\ref{eq:gamma_related_problem}), (\ref{eq:URLLC_subproblem}), (\ref{eq:UE_association_problem}), (\ref{eq:UAV_location_equal_problem_approximate}),
(\ref{eq:UAV_power_problem_approx}). As a closed-form solution is derived to solve (\ref{eq:gamma_related_problem}), its complexity is $O(7N^e)$ in the worst case. Binary search methods with the complexity of $O(\log_2(\frac{W^{\rm tot}}{\epsilon})+\log_2(\frac{W_{ub}^u(t)}{\epsilon}))$ are leveraged to solve (\ref{eq:URLLC_subproblem}) to obtain the optimal $W^u(t)$, where $\epsilon$ represents the error tolerance. The complexity of mitigating the linear integer programming problem (\ref{eq:UE_association_problem}) is $O((N^e+1)^{J})$. The complexities of solving convex problems (\ref{eq:UAV_location_equal_problem_approximate}),
(\ref{eq:UAV_power_problem_approx}) are $O((J + 4N^e)^{3.5})$, $O((J + N^e)^{3.5})$, respectively. Besides, (\ref{eq:UE_association_problem}), (\ref{eq:UAV_location_equal_problem_approximate}), (\ref{eq:UAV_power_problem_approx}) must be alternatively solved until converge or reach the maximum iteration times. Therefore, the complexity of solving (\ref{eq:subproblem_BPX}) is $O(f_2) = O(\log_2(\frac{W^{\rm tot}}{\epsilon}) + \log_2(\frac{W_{ub}^u(t)}{\epsilon}) + \bar r_{\rm max}( (N^e+1)^{J} + (J + 4N^e)^{3.5} + (J + N^e)^{3.5}))$ in the worst case. In summary, the total complexity of Algorithm \ref{alg:final_algorithm} is $O(T(f_1+7N^e+f_2))$ in the worst case. The computational complexity of the total complexity is exponential to the number of UAVs that is small. Moreover, the actual complexity is usually much smaller than the worst case.


\section{Simulation Results}

\subsection{Comparison algorithms and parameter setting}
To verify the effectiveness of the proposed algorithm, we compare it with two benchmark algorithms: 1) \textbf{Static UAV-based (SUAV) algorithm:} The difference between SUAV and RE$^2$FS lies in the scheme of controlling the movement of UAVs. For SUAV, it randomly deploys $J$ hovering UAVs with the similar deployment altitude (50 m) over the area of interest;
2) \textbf{CirCular Trajectory-based (CCT) algorithm:} In this algorithm, each UAV flies in a circular trajectory with a speed of 10 m/s. At the beginning of the simulation, UAVs (with an altitude of 50 m) are deployed in a line. The distance between two adjacent UAVs is $1/(2N)$ km. The horizontal locations of the first and the last UAVs are $[1/2+1/ (4N); 1/2]$ km and $[1-1/ (4N); 1/2]$ km, respectively, and turning radiuses of them are $1/ {(4N)} $ km and $1/2-1/ (4N) $ km. Besides, it adopts the similar slice request acceptance and UAV transmit power optimization schemes as RE$^2$FS.

We consider an urban area of size $1 \times 1$ km$^2$ with highrise buildings in the simulation. This scenario corresponds to the most challenging environment for slicing the UAV network, since the LoS/NLoS links may alter frequently as UAVs fly.
To accurately simulate the UtG and BtU channel gains in the environment, we generate the building locations and heights based on a realization of a local building model suggested by International Telecommunication Union (ITU) \cite{Propagation2012ITU} with statistical parameters $\alpha = 0.3$, $\beta = 300$ buildings/km$^2$, $\gamma$ being modelled as a Rayleigh distribution with the mean value $\sigma = 30$ m. The heights of all buildings are clipped to not exceed 40 m for convenience.
The BS antenna model follows the 3GPP specification \cite{Study20173GPP}, where an eight-element uniform linear array is placed vertically. Each array element is directional with half-power beamwidths along both vertical and horizontal dimensions equaling to $65^{\circ}$.
To simulate the signal strength measured by UAVs, the presence/absence of LoS link between a UAV and a ground user is firstly checked based on the building realization.
Meanwhile, we determine whether there exists an LoS link between the BS and a UAV to simulate the signal strength measured by the BS. Then, we generate the BtU and UtG path losses using the 3GPP model for urban Macro \cite{Technical20173GPP}.
The small-scale fading coefficient is added assuming Rayleigh fading for the NLoS case and Rician fading with 15 dB Rician factor for the LoS case \cite{zeng2020simultaneous}.

To test the practicality of the distributed ESN learning method, the realistic pedestrian movement dataset is extracted from a Github website\footnote{https://github.com/pswf/Twitter-Dataset/blob/master/Dataset. Our algorithm accommodates other realistic pedestrian movement datasets.} and utilized in the simulation.
The dataset depends on 12000 pieces of twitter information collected near Oxford street, in London on the 14th, March 2018. In this dataset, GPS-related position information of $N^e$ mobile users who tweeted more than two times were recorded.
Besides, to obtain more user information to describe users' movement more specifically, a linear interpolation method was used to make sure that the position information of each user was recorded every 200 seconds. After that, the 2D trajectory of each user was linearly zoomed into the simulated area of size $1 \times 1$ km$^2$.
In this case, the trajectory of each user was obtained. A turntable game in \cite{yang20193} was also used to set the required data rates of ground users with $C_s^{\rm th} \in \{1, 2, 4\}$ Mbps.

Additionally, the parameters related to URLLC slices are listed as follows: we consider one URLLC slice and set $\tau_{s,u}^{req} = 5$ milliseconds, $\epsilon_{s,u}^{req} = 1e$-7, $b_{s,u}^{req} = 160$ bits, $p_{\rm B}^{\rm max} = 50000$ mW, the user height $g_{i,s} = 1.8$ m, $\forall i, s$.
MBB slice related parameters are shown as below: we consider three types of MBB slices and set the UAV altitude $g_j(t) = 50$ m, the circuit power $p_j^c = 20$ mW, $\hat p_j = 1650$ mW, $\tilde p_j = 1500$ mW, $\forall j$. Besides, let $e_{\max} = 50$ m, $d_{\min} = 5$ m, $u_{i,s}^{\rm max}(t)$ is approximated as $W^{\rm tot}\log_2(1 + (\hat{p}_j - p_j^c)\theta_{ij,s}(t)/(N_0W^{\rm tot} |g_j(t)-g_{i,s}|^2))$.
Set other learning-correlated parameters as below: $r_{\rm max} = 100$, $\bar r_{\rm max} = 1000$, the sample number $Q = 6$, the number of future time slots $K = 10$, $N_i = 2$, $N_o = 2$, $N_r = 300$, $\lambda = 0.001$, the step size $\eta = 0.01$, $\xi = 0.001$. For each DNN, its first hidden layer has 512 neurons, and its second hidden layer has 256 neurons. The learning rate of each DNN is $0.001$, the minibatch size $|{\mathcal T}_t| = 64$, $\forall t$, the buffer capacity $C = 1e$+6.
More system parameters are listed as follows: $T_p = 1$, the carrier frequency {{$f_c = 2.0$} GHz}, light of speed {$c = 3.0e$+8 m/s}, $G_j = G_r = 1$ dBi, $\forall j$, total bandwidth $W^{\rm tot} = 10$ MHz, noise power spectral density $N_0 = -235$ dBm/Hz, {$T = 500$}, the coefficients $\rho = 0.01$ and $V = 2$, and the 3D location of the BS is $\bm x_{\rm B}^{\rm 3D} = [25; 37.5; 25]$ m.

\subsection{Performance evaluation}
To comprehensively understand the accuracy and the availability of the developed learning methods and optimization framework, we illustrate the performance results of the distributed learning method, online channel gain coefficient learning methods, Lyapunov-based optimization framework, respectively. In this simulation, we first let the UAV number $J = 3$ and the mobile user number $N^e = 64$.

\begin{figure}[!t]
\centering
\subfigure[Comparison of x-coordinate]{\includegraphics[width=2.3 in]{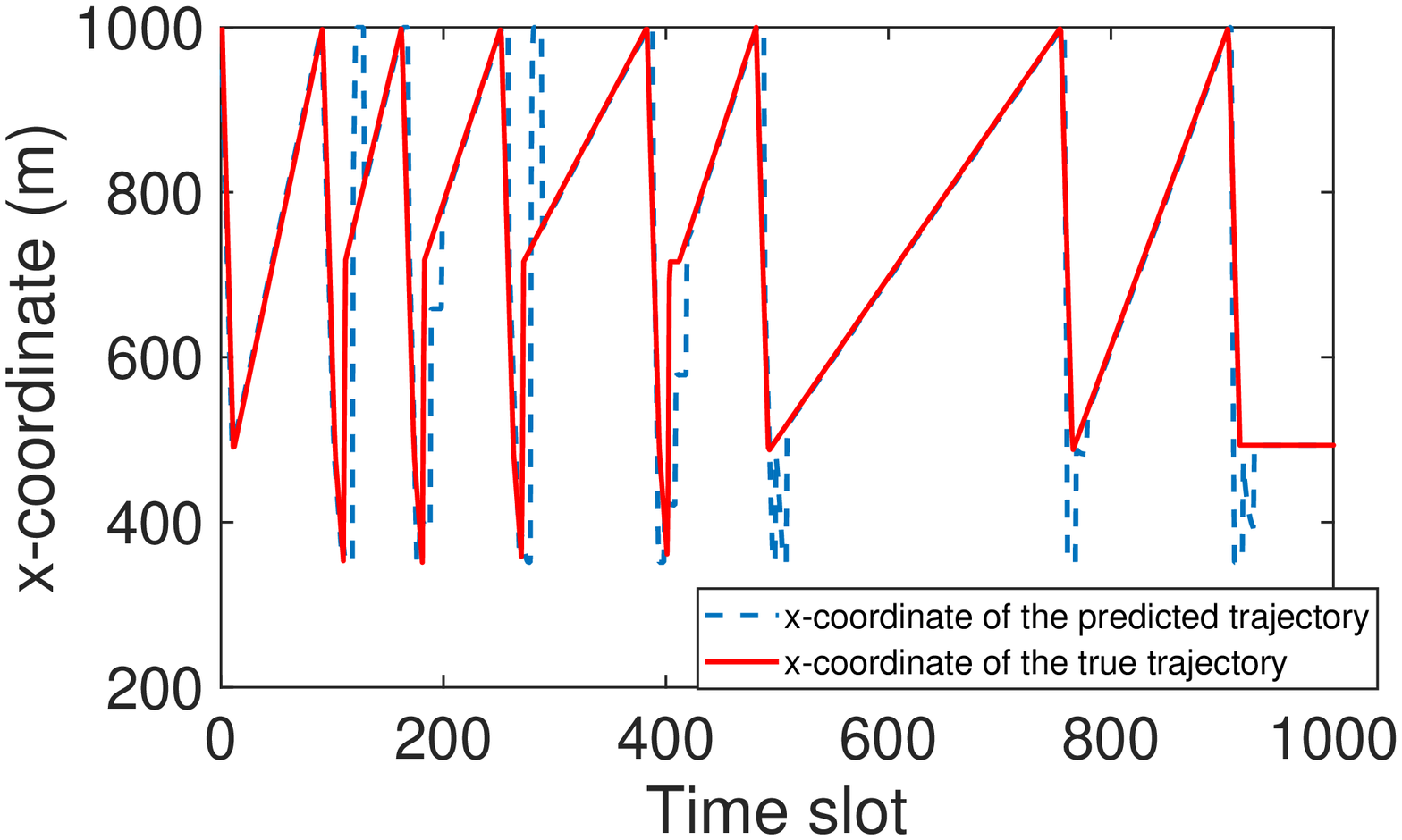}%
\label{fig_x_coordinate}}
\hfil
\subfigure[Comparison of y-coordinate]{\includegraphics[width=2.3 in]{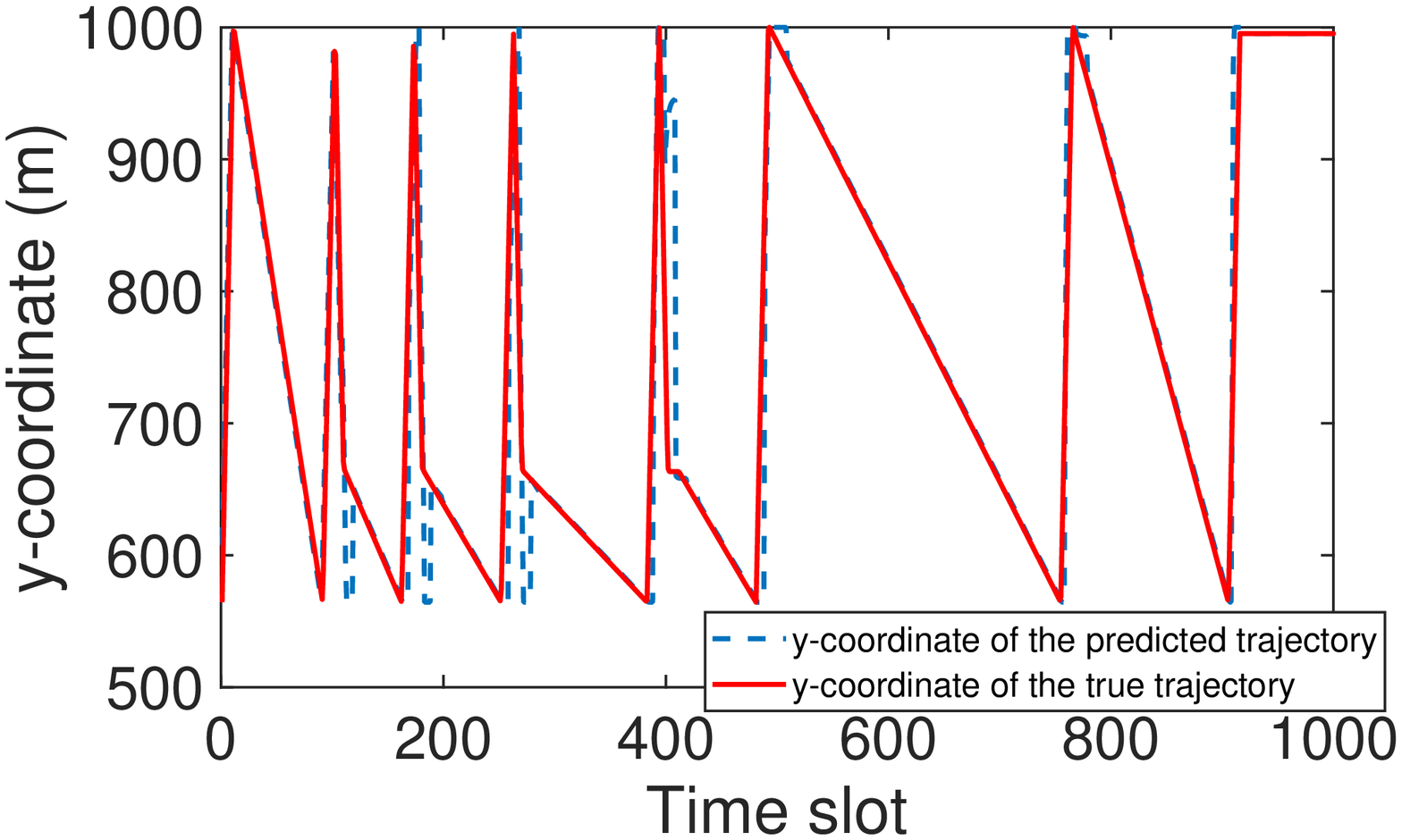}%
\label{fig_y_coordinate}}
\caption{Comparison of true and predicted trajectories of a user.}
\label{fig_x_y_coordinate}
\end{figure}
\begin{figure}[!t]
\centering
\includegraphics[width=3.0 in]{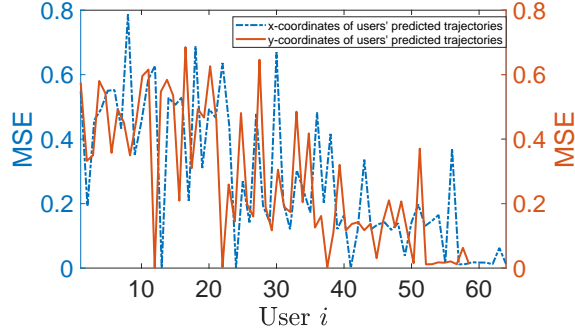}
\caption{Prediction accuracy of the distributed ESN learning method.}
\label{fig_esn_accuracy}
\end{figure}

First, to validate the accuracy of the distributed learning method on predicting users' locations, we plot the actual trajectory of a randomly selected user and its correspondingly predicted trajectory in Fig. \ref{fig_x_y_coordinate}. The accuracy, which is measured by the mean square error (MSE), of predicted trajectories of 64 mobile users is plotted in Fig. \ref{fig_esn_accuracy}.
From these figures, we can observe that:
1) when the heading directions of users will not change fast, this method can exactly predict their locations. When users change their moving directions quickly, the method loses their future locations. However, the method will re-capture the future locations of users after training ESN models based on newly collected users' location samples;
2) the obtained MSE of the predicted trajectories and actual trajectories of 64 mobile users will not be greater than $0.8$. Therefore, we may conclude that the developed distributed learning method can be utilized to predict users' locations.
\begin{figure}[!t]
\centering
\begin{minipage}[t]{0.45\textwidth}
\centering
\includegraphics[width=3.0 in]{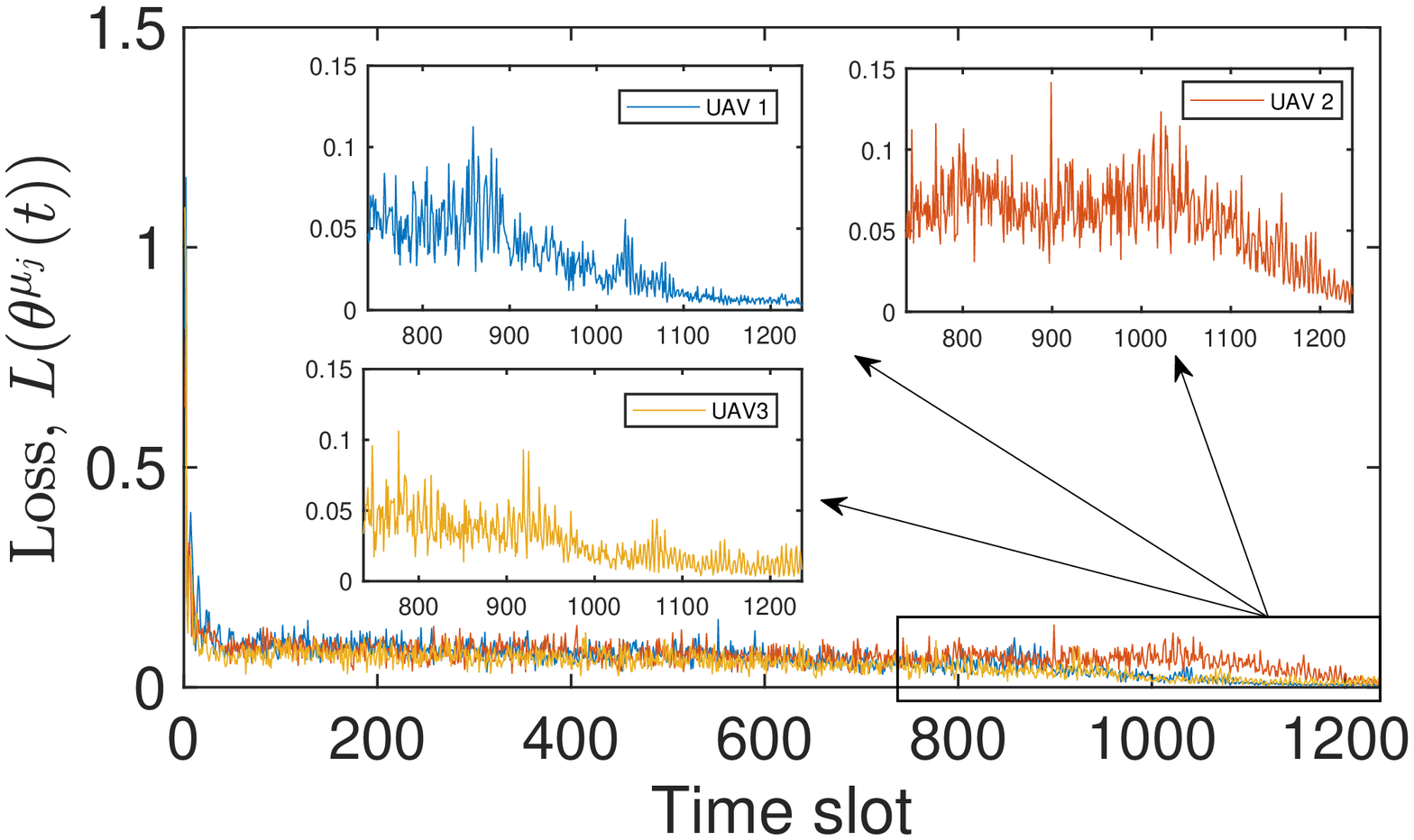}
\caption{Loss values of DNNs for UtG channel gain coefficient estimation versus time slot.}
\label{fig_downlink_DNN}
\end{minipage}
\hspace{0.05\linewidth}
\begin{minipage}[t]{0.45\textwidth}
\centering
\includegraphics[width=3.0 in]{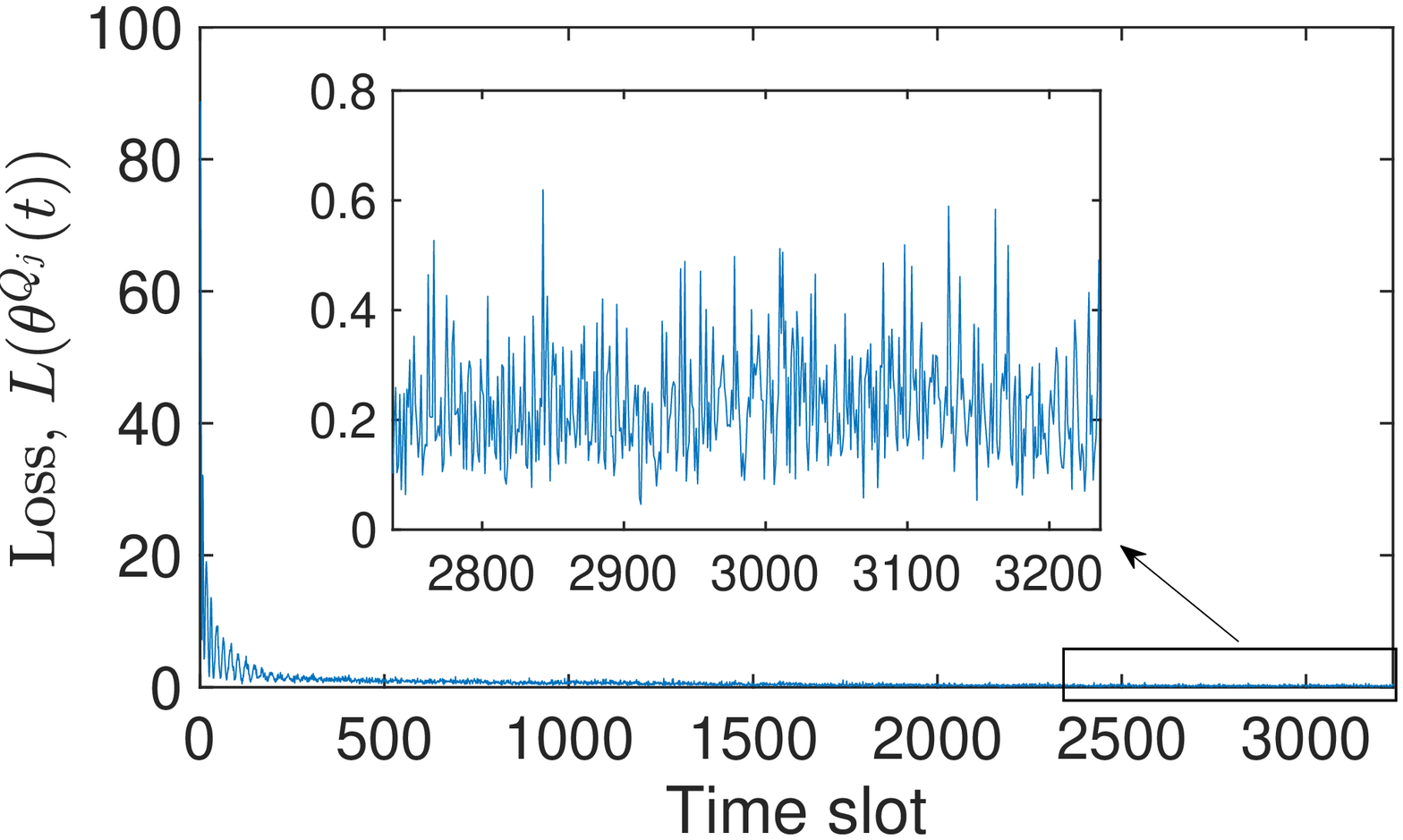}
\caption{Loss values of the DNN for BtU channel gain coefficient estimation versus time slot.}
\label{fig_uplink_DNN}
\end{minipage}
\end{figure}

Next, to testify the accuracy of online channel gain coefficient estimation methods, we plot the corresponding tendency of loss, which is calculated by (\ref{eq:mse}), between the estimated coefficient values and its target coefficient values in Figs. \ref{fig_downlink_DNN} and \ref{fig_uplink_DNN}.
Fig. \ref{fig_downlink_DNN} shows the loss values of the DNNs for UtG channel gain coefficient estimation. In this figure, loss values in the first 736 time slots, where initial values of 264 time slots are forgotten to alleviate the impact of noise, reflect the pre-training accuracy, and loss value in the later 500 time slots reflect the online learning accuracy.
Fig. \ref{fig_uplink_DNN} illustrates the loss results of the DNN for BtU channel gain coefficient estimation between the BS and the first UAV. Similarly, the values in the first 2736 time slots and the later 500 time slots reflect the pre-training and online learning accuracy.
From these figures, we can observe that:
1) the estimation error is great initially, but it quickly decreases with the increase of time slots as more experience is accumulated. For example, the obtained loss value is smaller than 0.2 after twenty time slots.
Besides, after a number of time slots, DNNs for BtU and UtG channel gain coefficient estimation can converge;
2) although actual BtU and UtG channel gain coefficients vary fast due to the movement of users and UAVs, the estimation method can achieve good estimation results. For example, during the online learning period, the loss values of DNNs for U2G channel gain coefficient estimation reach an order of less than $1.5e$-1, and the loss value of the DNN for BtU channel gain coefficient estimation is smaller than $0.65$.

Third, to verify the availability of the Lyapunov-based optimization framework, we plot the tendency of the virtual queue stability values, defined as ${S_Q} = \mathop {\max }\nolimits_{i,s} \  {[{Q_{i,s}}(t)]^ + }/t$, ${S_Z} = \mathop {\max }\nolimits_{i,s} \  {[{Z_{i,s}}(t)]^ + }/t$, and ${S_H} = \mathop {\max }\nolimits_{j \in {\mathcal J}} \  {[{H_j}(t)]^ + }/t$ in Fig. \ref{fig_queue_stability}. Besides, the trajectories of three UAVs in the first 50 time slots and their final 2D locations are plotted in Fig. \ref{fig_uav_track}.
\begin{figure}[!t]
\centering
\begin{minipage}[t]{0.45\textwidth}
\centering
\includegraphics[width=2.5 in]{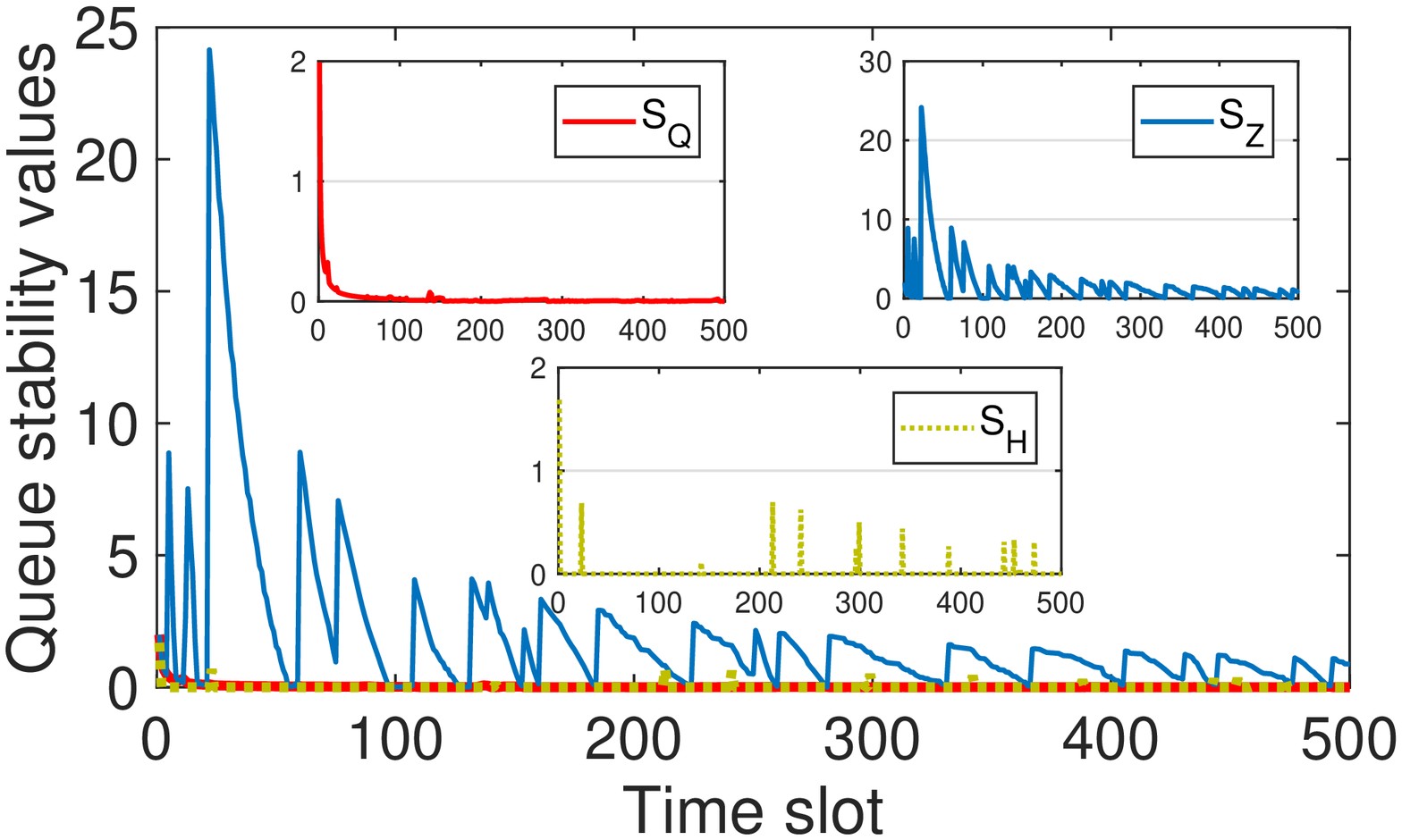}
\caption{Trend of virtual queue stability values versus time slot.}
\label{fig_queue_stability}
\end{minipage}
\hspace{0.05\linewidth}
\begin{minipage}[t]{0.45\textwidth}
\centering
\includegraphics[width=2.5 in]{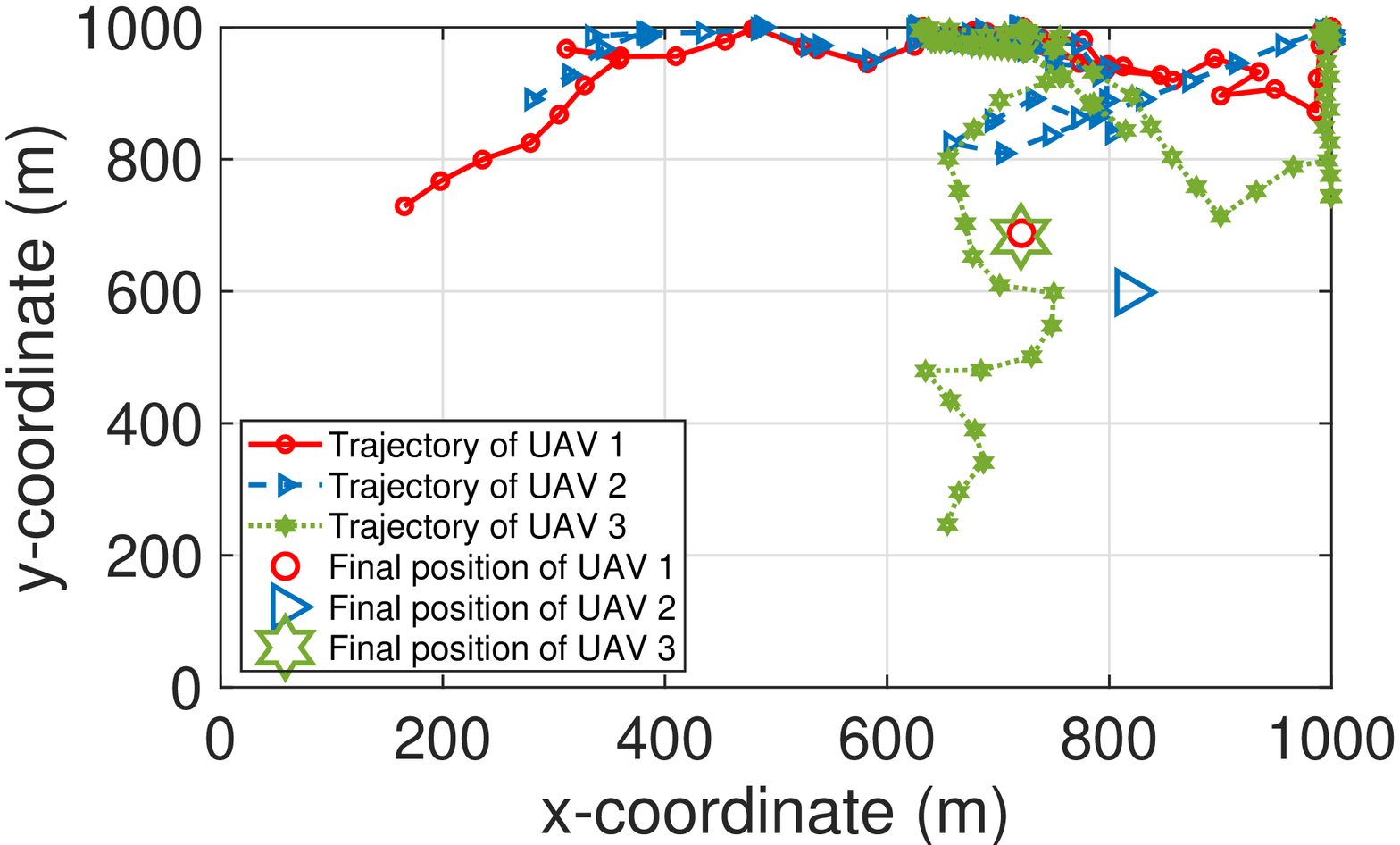}
\caption{Trajectories of three UAVs projected in a 2D space.}
\label{fig_uav_track}
\end{minipage}
\end{figure}

The following observations can be achieved from these figures: 1) during the learning period, the obtained queue stability values are bounded;
2) the obtained queue stability values tend to zero with an increasing time slot; as a result, all introduced virtual queues are mean-rate stable according to (\ref{eq:Queue_EQ}), i.e., all time average constraints in (\ref{eq:Jensen_problem}) can be satisfied. This result verifies the effectiveness of the Lyapunov-based optimization framework;
3) since the value of $S_Q$ tends to zero, the rate requirements of served users can be satisfied, which means that $W^e(t)$ is non-zero. If URLLC requirements of UAV control and non-payload information delivery are not satisfied, then $W^{\rm tot}$ will be allocated to URLLC slices. In this case, all MBB slices will be released and $S_Q$ will be monotonously increase with $t$, which is not shown in Fig. \ref{fig_queue_stability}. Therefore, URLLC requirements of UAV control and non-payload information transmission can be accommodated;
4) UAV movement constraints can be met at each time slot;
5) as users frequently appear in the upper right corner of the considered area, UAVs tend to move to this corner. In this way, QoS requirements of ground users can be met while the UAV transmit power can be reduced.

Next, we proceed to the verification of the effectiveness of the proposed RE$^2$FS algorithm by comparing it with other two benchmark algorithms.
To measure the effectiveness, the following two key performance indicators are introduced: the energy efficiency that is computed by (\ref{eq:original_problem}a) and the Jain's fairness index, defined as ${{{\left( \sum\nolimits_{i,s}{{{{\bar{u}}}_{i,s}}} \right)}^{2}}}/{N^e\sum\nolimits_{i,s}{\bar{u}_{i,s}^{2}}}$ with ${{\bar u}_{i,s}} = \frac{1}{T}\sum\nolimits_{t = 1}^T {{u_{i,s}}(t)} $.
\begin{figure}[!t]
\centering
\begin{minipage}[t]{0.45\textwidth}
\centering
\includegraphics[width=2.3 in]{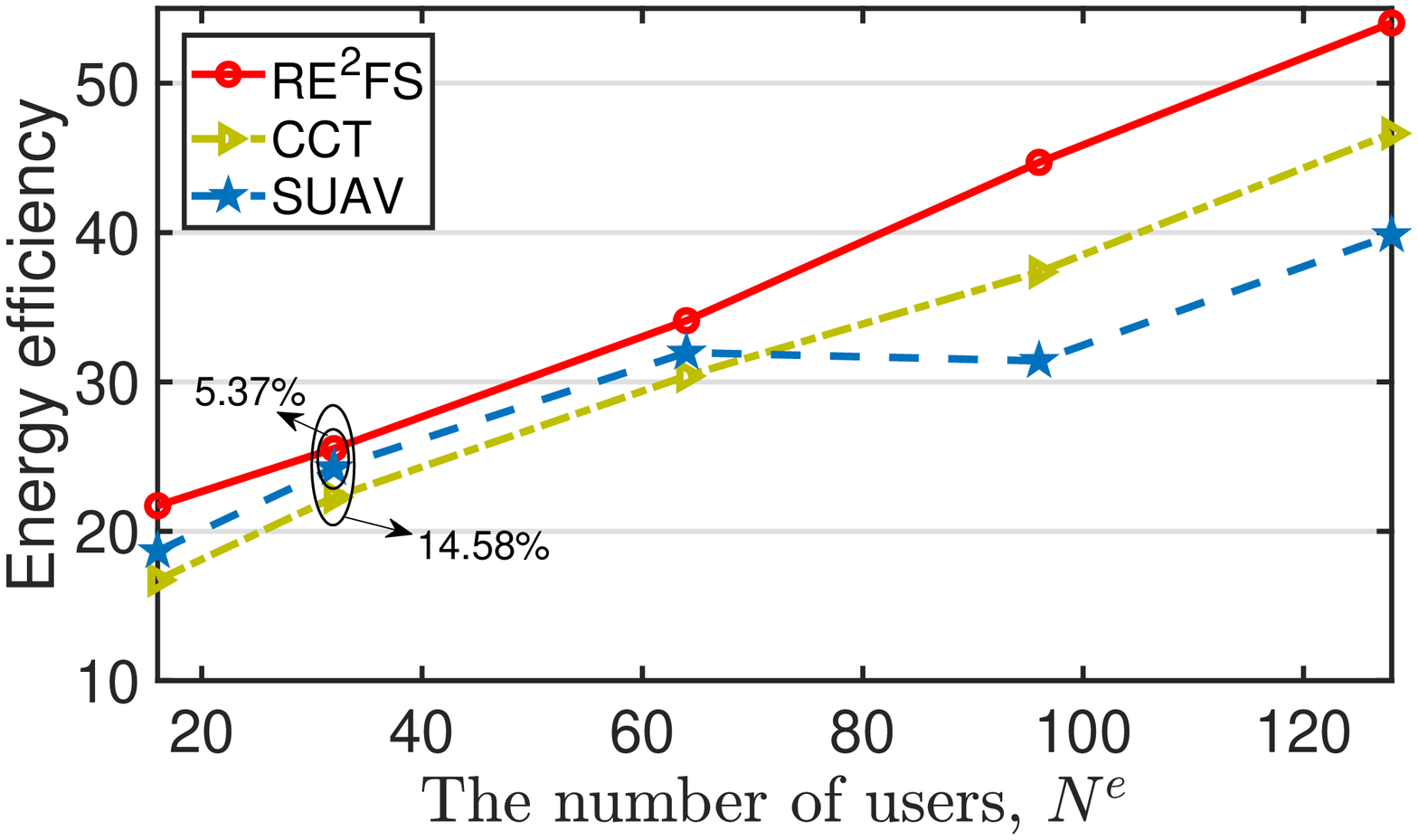}
\caption{Energy efficiency versus the number of users.}
\label{fig_EE_user_num}
\end{minipage}
\hspace{0.05\linewidth}
\begin{minipage}[t]{0.45\textwidth}
\centering
\includegraphics[width=2.3 in]{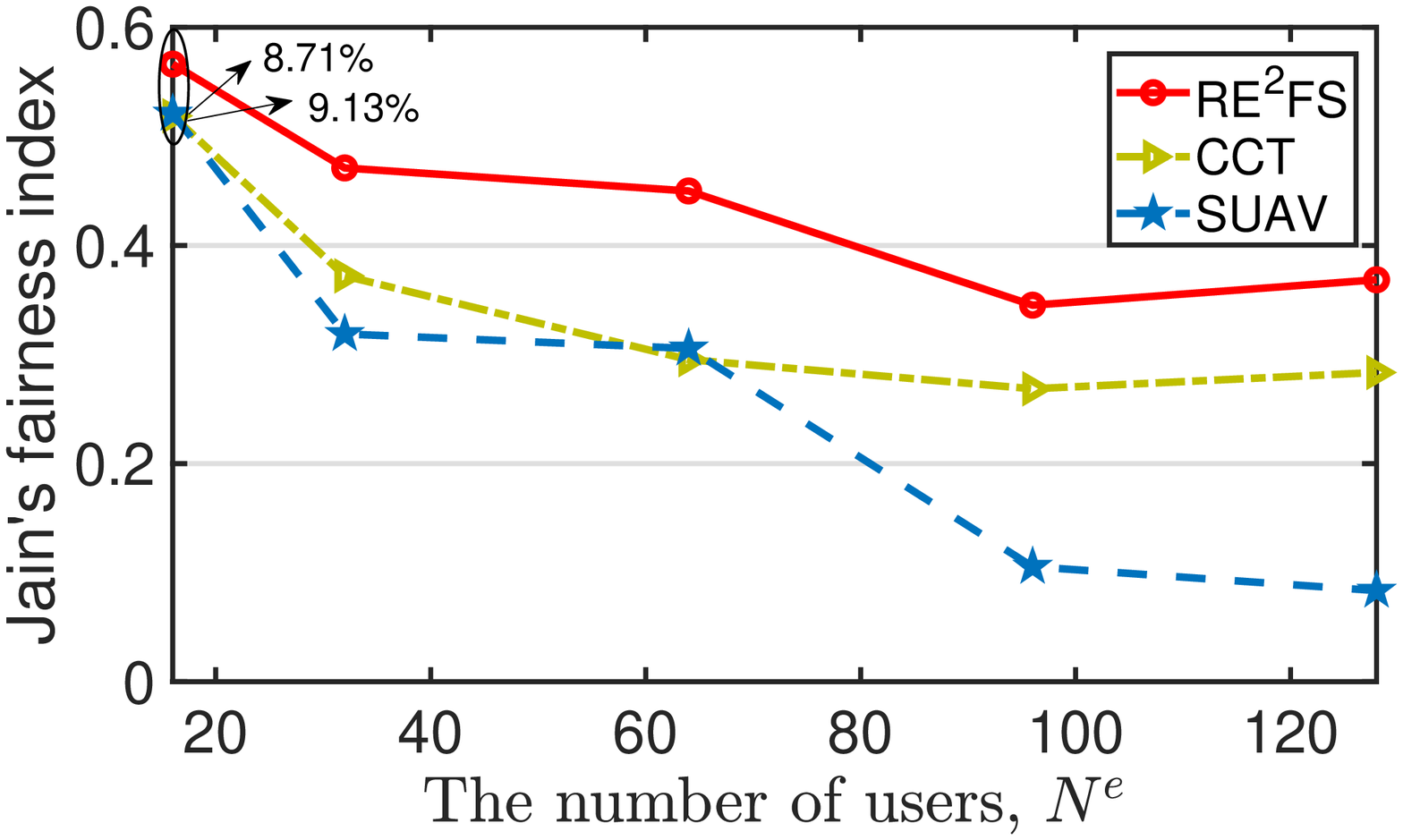}
\caption{Jain's fairness index versus the number of users.}
\label{fig_fairness_user_num}
\end{minipage}
\end{figure}

We first plot the impact of the number of users on the obtained energy efficiency and the Jain's fairness index of all comparison algorithms. Figs. \ref{fig_EE_user_num} and \ref{fig_fairness_user_num} show the tendency of the obtained energy efficiency and the Jain's fairness indexes of all algorithms, respectively, when $J = 3$, $N^e \in \{16, 32, 64, 96, 128\}$.

From these figures, we can observe that: 1) the proposed RE$^2$FS achieves the highest energy efficiency. For example, compared with the SUAV and CCT algorithms, RE$^2$FS improves the energy efficiency by $5.37\%$ and $14.58\%$, respectively, when $N^e = 32$;
2) when $N^e \le 64$, the achieved energy efficiency of SUAV is greater than that of CCT. However, CCT performs better in terms of the energy efficiency than SUAV when $N^e > 64$. The communication coverage of a UAV is limited. When the number of users is small (e.g., not more than 64), regular UAV trajectories may generate more coverage holes. On the contrary, mobile UAVs can serve more users when the number of users is great (e.g., greater than 64). This reason can also be utilized to explain the result that RE$^2$FS can significantly improve the energy efficiency when $N^e > 64$;
3) the achieved energy efficiency values of all comparison algorithms will increase with an increasing number of users as more users can be served;
4) the proposed RE$^2$FS achieves the highest Jain's fairness index. For instance, when $N^e = 16$, compared with the SUAV and CCT algorithms, RE$^2$FS improves the energy efficiency by $8.71\%$ and $9.13\%$, respectively;
5) when $N^e > 64$, owing to the movement of UAVs, CCT achieves greater fairness index than SUAV;
6) as $N^e$ is part of the denominator of the definition of the Jain's fairness index, the obtained fairness indexes of all comparison algorithms with $N^e = 16$ are greater than their achieved fairness indexes when $N^e = 128$.

We then illustrate the impact of the number of UAVs on the energy efficiency and Jain's fairness indexes of all comparison algorithms. Figs. \ref{fig_EE_uav_num} and \ref{fig_fairness_uav_num} show the tendency of the obtained performance indicators when $N^e = 64$.
\begin{figure}[!t]
\centering
\begin{minipage}[t]{0.45\textwidth}
\centering
\includegraphics[width=2.3 in]{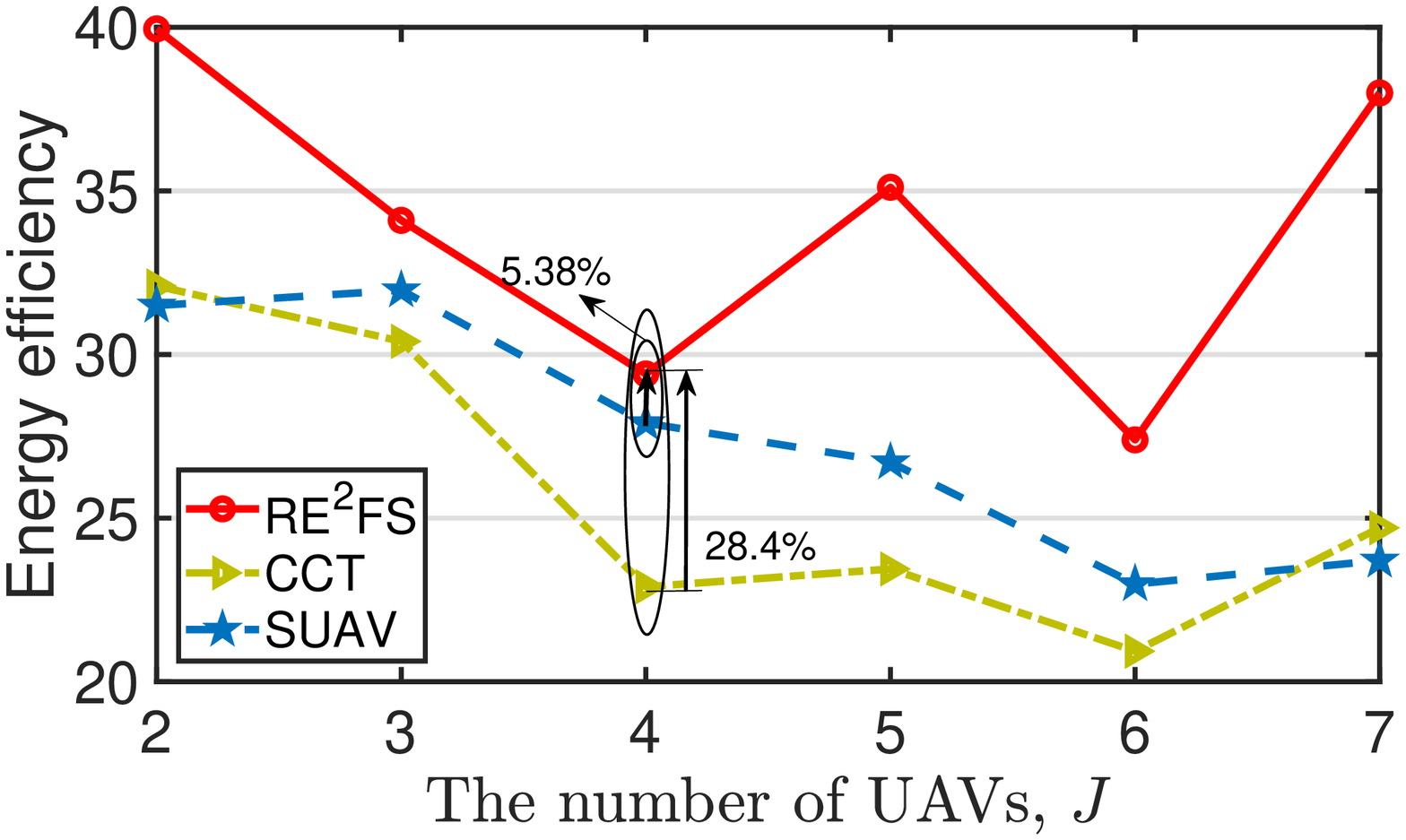}
\caption{Energy efficiency versus the number of UAVs.}
\label{fig_EE_uav_num}
\end{minipage}
\hspace{0.05\linewidth}
\begin{minipage}[t]{0.45\textwidth}
\centering
\includegraphics[width=2.3 in]{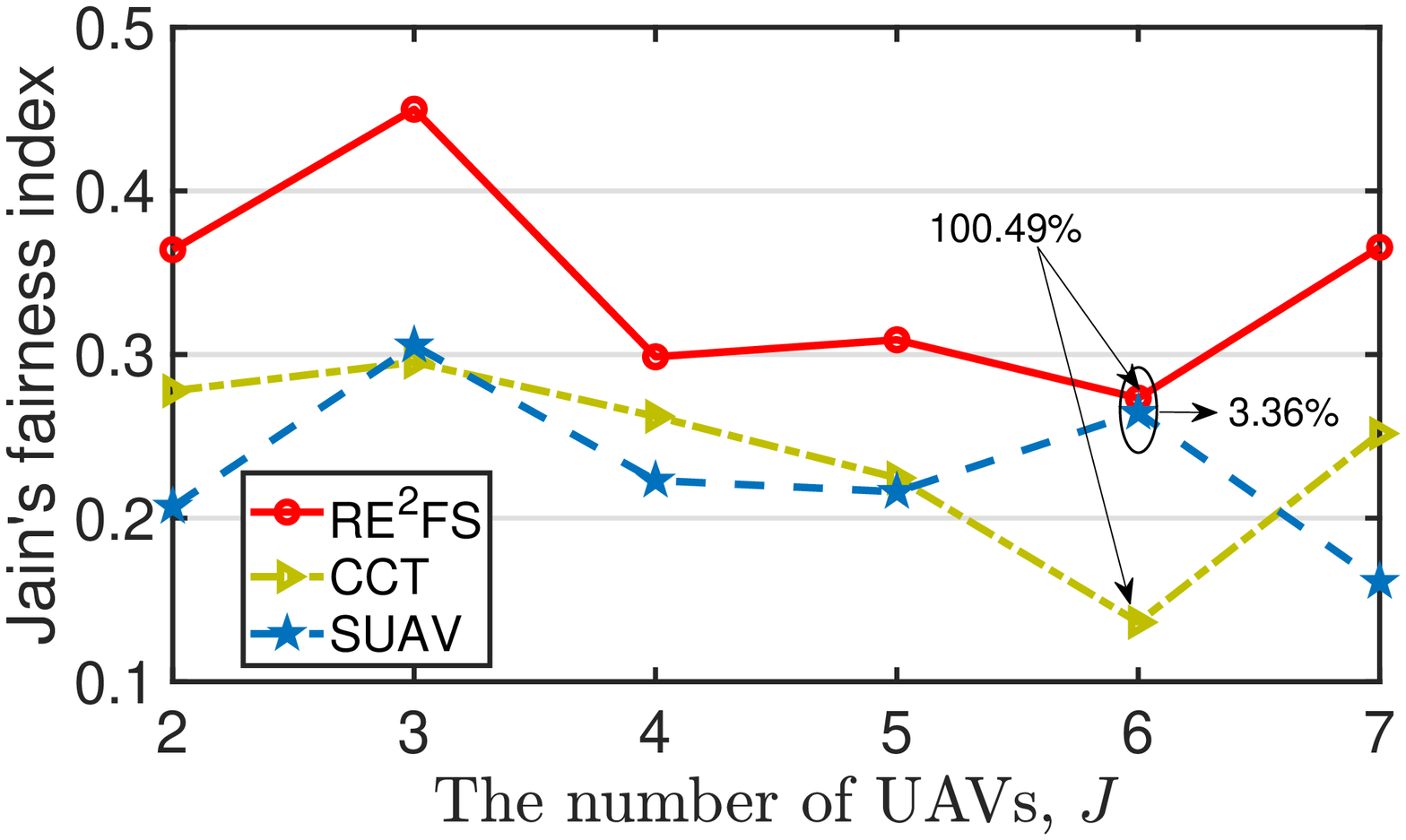}
\caption{Jain's fairness index versus the number of UAVs.}
\label{fig_fairness_uav_num}
\end{minipage}
\end{figure}

From these figures, we can observe that: 1) owing to the effective UAV movement control, the proposed RE$^2$FS achieves the highest energy efficiency and Jain's fairness index. For instance, compared with the benchmark algorithms, the minimum improvement of RE$^2$FS on the energy efficiency and Jain's fairness index is $5.38\%$ and $3.36\%$, respectively;
2) we cannot conclude that the obtained performance indicators of all comparison algorithms will increase or decrease with an increasing number of UAVs. More mobile users can be simultaneously served when the number of UAVs is increased, which indicates that the energy efficiency value and the fairness index may be raised. However, more UAVs will result in greater interference and consume more energy, which will decrease the energy efficiency value; meanwhile, owing to the strong interference some mobile users may experience coverage interruptions, which will lower the fairness index;
3) when varying the number of UAVs, the performance of the SUAV algorithm, where UAVs are hovering, does not outperform the CCT algorithm, where UAVs follow circular trajectories. Summarily, the above results show that the URLLC requirements of UAV control and non-payload information delivery can be accommodated and the UAV network can provide energy-efficient and fair MBB services for ground mobile users by exploiting the proposed RE$^2$FS algorithm.
\section{Conclusion}
This paper investigated a proactive UAV network slicing problem and formulated the problem as a sequential decision problem with a goal of providing energy-efficient and fair service coverage for MBB users while satisfying the URLLC requirements of UAV control and non-payload signal delivery.
This problem was confirmed to be a mixed-integer-non-convex optimization problem, the solution of which also required accurate mobile users' locations and channel gain models.
To solve this problem, we proposed a new approach using learning and optimization methods.
Specifically, we first developed a distributed learning method to predict mobile users' locations, with which we built analytically tractable DNN-based channel gain models. Then, we proposed a Lyapunov-based optimization framework to decompose the original problem into several repeated optimization subproblems based on the learned results. Finally, these subproblems were optimized by exploiting an SCA technique and an alternative optimization scheme.
Simulation results were provided to show the accuracy of the learning methods and to verify the effectiveness of the Lyapunov-based optimization framework.

\appendix
\subsection{Proof of Lemma \ref{lem:ADMM_ESN}}
For all $j \in {\mathcal J}$, to prove that ${{\mathcal L}(\bm W_{j,t}^{(r)}, \hat {\bm W}_t^{(r)})} $ is bounded, we should prove that the matrices $\bm W_{j,t}^{(r)}$ and $\hat {\bm W}_t^{(r)}$ are bounded. Next, we should prove that there exist non-positive coefficients $\varepsilon _{j}$ and $\hat \varepsilon _{j}$ such that $|{{\mathcal L}^{(r+1)}(\cdot)}  -  {{\mathcal L}^{(r)}(\cdot)}| \le  {\varepsilon_{j}}|| \bm W_{j,t}^{(r+1)} - \bm W_{j,t}^{(r)} ||_2  +  {{\hat \varepsilon_{j}}|| {\hat {\bm W}_t^{(r+1)} - \hat {\bm W}_t^{(r)}} ||_2} $. Since the ADMM is exploited to solve (\ref{eq:least_square_prob}), the convergence of the distributed ESN method is determined by that of the ADMM. Then the convergence of the ADMM should be proved. However, we omit the detail proof as it is able to be found in the convergence proof of the ADMM in \cite{boyd2011distributed,wang2019global}. This completes the proof.

\subsection{Proof of Lemma \ref{lem:lemma_equivalent}}
By referring to the Jensen's inequality we can achieve $\bar{g}(t)\le \phi ({{\bar{\bm \gamma }}}(t))$. Next, suppose all limits in (\ref{eq:Jensen_problem}) exist, we have ${{\bar \gamma }_{i,s}}(t) \le {{\bar u}_{i,s}}(t)$. Besides, as $ \phi(\bar {\bm \gamma}(t))$ monotonously increases with ${{\bar \gamma }_{i,s}}(t)$ $\forall i,s$, we can obtain $\bar g(t) \le \phi ({{\bar u}_{1,1}}(t), \ldots ,{{\bar u}_{N_{|{\mathcal S}^e|}^e,|{\mathcal S}^e|}}(t))$.
It means that the maximum value of the objective function of (\ref{eq:Jensen_problem}) is not greater than that of (\ref{eq:original_problem}). Besides, the maximum value of the objective function of (\ref{eq:original_problem}) can be obtained through letting ${{\bar \gamma }_{i,s}}(t) = {\bar u}_{i,s}^{\star}(t)$ for all $i \in {\mathcal N}_s^e$, $s \in {\mathcal S}^e$ and $t \in \{1, 2, \ldots\}$ with $({\bar u}_{1,1}^{\star}(t), \ldots, {\bar u}_{N_{|{\mathcal S}^e|}^e, |{\mathcal S}^e|}^{\star}(t))$ being the optimal time average achievable data rates of all MBB users for (\ref{eq:original_problem}). Therefore, (\ref{eq:Jensen_problem}) and (\ref{eq:original_problem}) are equivalent.

We can observe that (\ref{eq:Jensen_problem}) only includes time average terms; thus, a drift-plus-penalty technique \cite{Neely2014A} is explored to alleviate (\ref{eq:Jensen_problem}).
Specifically, to enforce the constraint (\ref{eq:Jensen_problem}c), we use (\ref{eq:Queue_ui}) to define the virtual queue ${{Q}_{i,s}}(t)$ for all $i$, $s$, $t$. Then we can conclude that the constraint (\ref{eq:Jensen_problem}c) is satisfied if the mean-rate stability condition (\ref{eq:Queue_EQ}) holds \cite{Neely2014A}.

Similarly, we define other two virtual queues $Z_{i,s}(t)$, ${{H}_{j}}(t)$ for all $i$, $s$ and $j$ using (\ref{eq:Queue_Z}) and (\ref{eq:Queue_H}) to enforce constraints (\ref{eq:Jensen_problem}b) and (\ref{eq:Jensen_problem}d), respectively. Then, the constraints (\ref{eq:Jensen_problem}b) and (\ref{eq:Jensen_problem}d) are satisfied if the corresponding mean-rate stability conditions in (\ref{eq:Queue_EQ}) hold, for all $i$, $s$, and $j$. 
This completes the proof.

\subsection{Proof of Lemma \ref{lemma:1}}
We discuss the upper bound of $\frac{1}{2}{({[ {Q_{i,s}(t + 1)} ]^ + })^2}$ in three cases. According to {(\ref{eq:Queue_ui})} and the non-negative operation,

Case 1: when ${Q_{i,s}(t + 1)}  \ge 0$ and ${Q_{i,s}(t)}  \ge 0$, we can obtain
\begin{equation}\label{eq:31}
\begin{array}{*{20}{l}}
{\frac{1}{2}{{({{[{Q_{i,s}}(t + 1)]}^ + })}^2}} = \frac{1}{2}{{({{[{Q_{i,s}}(t)]}^ + })}^2} +
{{[{Q_{i,s}}(t)]}^ + }(C_s^{th} - {u_{i,s}}(t))
{ + \frac{1}{2}{{(C_s^{th} - {u_{i,s}}(t))}^2}}.
\end{array}
\end{equation}

Case 2: when ${Q_{i,s}(t + 1)} \ge 0$ and ${Q_{i,s}(t)} < 0$, we can achieve $C_s^{th}-{{u}_{i,s}}(t)>{{Q}_{i,s}}(t+1)\ge 0$, ${{[ Q_{i,s}(t) ]}^{+}}=0$ and
\begin{equation}\label{eq:31}
\begin{array}{*{20}{l}}
{\frac{1}{2}{{({{[{Q_{i,s}}(t + 1)]}^ + })}^2} < \frac{1}{2}{{(C_s^{th} - {u_{i,s}}(t))}^2}}
{ = \frac{1}{2}{{({{[{Q_{i,s}}(t)]}^ + })}^2} + {{[{Q_{i,s}}(t)]}^ + }(C_s^{th} - {u_{i,s}}(t))}
{ + \frac{1}{2}{{(C_s^{th} - {u_{i,s}}(t))}^2}}.
\end{array}
\end{equation}

Case 3: when ${{Q_{i,s}(t + 1)} } < 0$, we can obtain
\begin{equation}\label{eq:31}
\begin{array}{*{20}{l}}
{\frac{1}{2}{{({{[{Q_{i,s}}(t + 1)]}^ + })}^2} = 0}\\
\qquad \qquad { \le \frac{1}{2}({{[{Q_{i,s}}(t)]}^ + } + {{(C_s^{th} - {u_{i,s}}(t))}^2}}\\
\qquad \qquad { = \frac{1}{2}{{({{[{Q_{i,s}}(t)]}^ + })}^2} + {{[{Q_{i,s}}(t)]}^ + }(C_s^{th} - {u_{i,s}}(t))}
{ + \frac{1}{2}{{(C_s^{th} - {u_{i,s}}(t))}^2}}.
\end{array}
\end{equation}

Therefore, we have
\begin{equation}\label{eq:L_Q}
\begin{array}{*{20}{l}}
{\frac{1}{2}{{({{[{Q_{i,s}}(t + 1)]}^ + })}^2} \le \frac{1}{2}{{({{[{Q_{i,s}}(t)]}^ + })}^2}}
{ + {{[{Q_{i,s}}(t)]}^ + }(C_s^{th} - {u_{i,s}}(t)) + \frac{1}{2}{{(C_s^{th} - {u_{i,s}}(t))}^2}}.
\end{array}
\end{equation}

Similarly, according to {(\ref{eq:Queue_Z}), (\ref{eq:Queue_H})} and the non-negative operation, we have
\begin{equation}\label{eq:L_Z}
\begin{array}{*{20}{l}}
\begin{array}{l}
\frac{1}{2}{({[{Z_{i,s}}(t + 1)]^+ })^2} = \frac{1}{2}{({[{Z_{i,s}}(t)]^+ })^2} +
\end{array}
{{{[{Z_{i,s}}(t)]}^+ }({\gamma _{i,s}}(t) - {u_{i,s}}(t)) + \frac{1}{2}{{({\gamma _{i,s}}(t) - {u_{i,s}}(t))}^2}},
\end{array}
\end{equation}
and
\begin{equation}\label{eq:L_H}
\begin{array}{*{20}{l}}
{\frac{1}{2}{{({{[{H_j}(t + 1)]}^ + })}^2} \le \frac{1}{2}{{({{[{H_j}(t)]}^ + })}^2}}
{ + {{[{H_j}(t)]}^ + }({p_j}(t) - {{\tilde p}_j} + p_j^c) + \frac{1}{2}{{({p_j}(t) - {{\tilde p}_j} + p_j^c)}^2}}.
\end{array}
\end{equation}

With inequalities (\ref{eq:L_Q})-(\ref{eq:L_H}), we can obtain a new inequality by utilizing the definition of Lyapunov drift. Next, we can achieve (\ref{eq:upper_bound}) by adding $-V\left( g(t)-\rho \sum\nolimits_{j=1}^{N}{{{p}_{j}^{\rm tot}}(t)} \right)$ to both sides of the new inequality. This completes the proof.

\subsection{Proof of Proposition \ref{lemma:lemma_uav_location}}
For any given UAV transmit power ${\mathcal P}(t)$, UAVs' locations ${{\mathcal {X}}(t-1)}$ at the previous time slot $t-1$, and the acceptance set of slice requests ${\mathcal A}(t)$, the variables ${{\mathcal {X}}(t)}$ in (\ref{eq:MBB_subproblem}) can be optimized via mitigating the following problem:
\begin{subequations}\label{eq:UAV_location_problem}
\begin{alignat}{2}
& \mathop {\rm Maximize }\limits_{{{\mathcal {X}}(t)}} {\mkern 1mu} \text{ } {\sum\nolimits_{i,s} {\{ {{[{Q_{i,s}}(t)]}^ + } + {{[{Z_{i,s}}(t)]^+} }\} {u_{i,s}}(t)} }  \\
&{\rm s.t: \quad } {{\rm (\ref{eq:original_problem}e),(\ref{eq:original_problem}f).}}
\end{alignat}
\end{subequations}

To simplify (\ref{eq:UAV_location_problem}a), we introduce slack variables $\{ \eta_{i,s}\}$, with which (\ref{eq:UAV_location_problem}) can be reformulated as
\begin{subequations}\label{eq:UAV_location_equal_problem}
\begin{alignat}{2}
& \mathop {\rm Maximize }\limits_{{{{\mathcal {X}}(t)}, {\{ \eta_{i,s}}(t) \} }} {\mkern 1mu} \text{ } {\sum\nolimits_{i,s} {\{ {{[{Q_{i,s}}(t)]}^ + } + {{[{Z_{i,s}}(t)]^+} }\} {\eta _{i,s}}(t)} }  \\
&{\rm s.t:} \quad {u_{i,s}}(t) \ge {\eta _{i,s}}(t),\forall i,s,t \\
& \qquad {\rm (\ref{eq:original_problem}e),(\ref{eq:original_problem}f).}
\end{alignat}
\end{subequations}

If $\eta_{i,s}^{\star}$ is the optimal solution to (\ref{eq:UAV_location_equal_problem}) such that the constraint (\ref{eq:UAV_location_equal_problem}b) is satisfied with strict inequality, we can then decrease $u_{i,s}(t)$ to make (\ref{eq:UAV_location_equal_problem}b) active without changing the value of (\ref{eq:UAV_location_equal_problem}a). Therefore, (\ref{eq:UAV_location_equal_problem}) is equivalent to (\ref{eq:UAV_location_problem}).

As ${h_{ij,s}}(t) = \frac{{{\theta _{ij,s}(t)}}}{{{|g_j(t)-g_{i,s}|^2} + ||{{\bm v}_j}(t) - {{\bm x}_{i,s}(t)}|{|_2^2}}}$, the achievable data rate of user $i$ can be expressed as ${u_{i,s}}(t) = \sum\nolimits_{j \in \mathcal J} {{a_{ij,s}}(t)W^e(t){R_{ij,s}}(t)}$ with
\begin{equation}\label{eq:R_ijs}
\begin{array}{l}
{{R_{ij,s}}(t) = }
{\log _2}\left( {1 + \frac{{{p_j}(t){\theta _{ij,s}}(t)/(|{g_j}(t) - {g_{i,s}}{|^2} + ||{\bm v_j}(t) - \bm {x_{i,s}}(t)||_2^2)}}{{{N_0}{W^e}(t) + \sum\limits_{k \in {\mathcal J}\backslash \{ j\} } {\frac{{{p_k}(t){\theta _{ik,s}}(t)}}{{|{g_k}(t) - {g_{i,s}}{|^2} + ||{\bm v_k}(t) - {\bm x_{i,s}}(t)||_2^2}}} }}} \right)\\
= {\log _2}\left( {{N_0}{W^e}(t) + \sum\limits_{k \in {\mathcal J}} {\frac{{{p_k}(t){\theta _{ik,s}}(t)}}{{|{g_k}(t) - {g_{i,s}}{|^2} + ||{\bm v_k}(t) - {\bm x_{i,s}}(t)||_2^2}}} } \right) -
{\log _2}\left( {{N_0}{W^e}(t) + \sum\limits_{k \in {\mathcal J}\backslash \{ j\} } {\frac{{{p_k}(t){\theta _{ik,s}}(t)}}{{|{g_k}(t) - {g_{i,s}}{|^2} + ||{\bm v_k}(t) - {\bm x_{i,s}}(t)||_2^2}}} } \right).
\end{array}
\end{equation}

(\ref{eq:UAV_location_equal_problem}) is not convex due to the non-convex constraints (\ref{eq:original_problem}f), and (\ref{eq:UAV_location_equal_problem}b). Therefore, we may not find efficient methods to obtain the optimal solution to (\ref{eq:UAV_location_equal_problem}).
Although (\ref{eq:UAV_location_equal_problem}b) is not concave with respect to (w.r.t) ${\bm v}_j(t)$, we can observe that the first term on the right-hand-side of (\ref{eq:R_ijs}) is convex w.r.t ${||{{\bm v}_k}(t) - {{\bm x}_{i,s}(t)}||_2^2}$, $\forall k \in {\mathcal J}$.
Accordingly, a slack variable ${{B}_{ik,s}(t)} = {||{{\bm v}_k}(t) - {{\bm x}_{i,s}(t)}||_2^2}$, $\forall i, k \ne j$ is involved to transform (\ref{eq:UAV_location_equal_problem}) into the following new problem
\begin{subequations}\label{eq:UAV_location_equal_problem_S}
\begin{alignat}{2}
& \mathop {\rm Maximize }\limits_{{\mathcal {X}}(t), \{ \eta_{i,s}(t) \}, \{ B_{ik,s}(t) \}} {\mkern 1mu} \text{ } {\sum\nolimits_{i,s} {\{ {{[{Q_{i,s}}(t)]}^ + } + {{[{Z_{i,s}}(t)]^+} }\} {\eta _{i,s}}(t)} }  \\
& {\rm s.t:} \text{ } \sum\nolimits_{j \in {\mathcal J}} {a_{ij,s}}(t)W^e(t)( {{\hat R}_{i,s}}(t) + {{{\tilde R}_{ij,s}}(t)})  \ge {\eta _{i,s}}(t), \forall i,s,t \\
& \quad {B_{ik,s}(t)} \le ||{{\bm v}_k}(t) - {{\bm x}_{i,s}(t)}||{^2},\forall i,s,k \ne j,t \\
& \quad {\rm (\ref{eq:original_problem}e),(\ref{eq:original_problem}f).}
\end{alignat}
\end{subequations}
where
\begin{equation}
{{\hat R}_{i,s}}(t) ={{{\log }_2}\left( {{N_0}{W^e}(t) + \sum\limits_{k \in {\cal J}} {\frac{{{p_k}(t){\theta _{ik,s}}(t)}}{{|{g_k}(t) - {g_{i,s}}{|^2} + {||{{\bm v}_k}(t) - {{\bm x}_{i,s}(t)}||_2^2}}}} } \right)},
\end{equation}
and
\begin{equation}
    {\tilde R_{ij,s}}(t) =  - {\log _2}\left( {{N_0W^e(t)} + \sum\limits_{k \in {\mathcal J}\backslash \{ j\} } {\frac{{{{p_k}(t)\theta _{ik,s}(t)}}}{{{|g_k(t)-g_{i,s}|^2} + {B_{ik,s}}(t)}}} } \right).
\end{equation}

Similar to (\ref{eq:UAV_location_equal_problem}), after introducing the slack variable $B_{ik,s}(t)$, (\ref{eq:UAV_location_equal_problem_S}) is equivalent to (\ref{eq:UAV_location_equal_problem}). Unfortunately, (\ref{eq:UAV_location_equal_problem_S}) is still non-convex as (\ref{eq:original_problem}f), (\ref{eq:UAV_location_equal_problem_S}b), and (\ref{eq:UAV_location_equal_problem_S}c) are non-convex.

To handle the non-convexity of (\ref{eq:UAV_location_equal_problem_S}), the SCA technique is explored. It can be observed that $\hat R_{i,s}(t)$, $\forall i,s,j$ is convex w.r.t $B_{ik,s}(t)$ and will be globally lower-bounded by its first-order Taylor expansion at any local point \cite{Stephen2004Convex}. Therefore, for a given local point at the $(r+1)$-th iteration ($r \ge 0$), denoted by ${{\bm v}_k^{(r)}(t)}$, $\hat R_{i,s}(t)$ is lower-bounded by
\begin{equation}\label{eq:Rij}
\begin{array}{l}
{{\hat R}_{i,s}}(t) \ge \tilde f_{i,s}(t) = {\log _2}\left( {{N_0W^e(t)} + \sum\nolimits_{k \in {\mathcal J}} {\frac{{{{p_k}(t)\theta _{ij,s}(t)}}}{{{|g_k(t)-g_{i,s}|^2} + ||{\bm v}_k^{(r)}(t) - {{\bm x}_{i,s}(t)}|{|_2^2}}}} } \right)\\
 - \sum\limits_{k \in {\mathcal J}} {\frac{{\frac{{{{p_k}(t)\theta _{ij,s}(t)}}}{{{{\left( {{|g_k(t)-g_{i,s}|^2} + ||{\bm v}_k^{(r)}(t) - {{\bm x}_{i,s}(t)}|{|_2^2}} \right)}^2}}}\left( {||{{\bm v}_k}(t) - {{\bm x}_{i,s}(t)}|{|_2^2} - ||{\bm v}_k^{(r)}(t) - {{\bm x}_{i,s}(t)}|{|_2^2}} \right)}}{{\left( {{N_0W^e(t)} + \sum\limits_{k \in {\mathcal J}} {\frac{{{{p_k}(t)\theta _{ij,s}(t)}}}{{{|g_k(t)-g_{i,s}|^2} + ||{\bm v}_k^{(r)}(t) - {{\bm x}_{i,s}(t)}|{|_2^2}}}} } \right)\ln 2}}} \\
 = D_{i,s}^{(r)}(t) - \sum\limits_{k \in {\mathcal J}} {E_{ik,s}^{(r)}(t)( {||{{\bm v}_k}(t) - {{\bm x}_{i,s}(t)}|{|_2^2} -}} {{||{\bm v}_k^{(r)}(t) - {{\bm x}_{i,s}(t)}|{|_2^2}} )},
\end{array}
\end{equation}
where
\begin{equation}
    D_{i,s}^{(r)}(t) = {\log _2}\left( {{N_0W^e(t)} + \sum\limits_{k \in {\mathcal J}} {\frac{{{{p_k}(t)\theta _{ij,s}(t)}}}{{{|g_k(t)-g_{i,s}|^2} + ||{\bm v}_k^{(r)}(t) - {{\bm x}_{i,s}(t)}|{|_2^2}}}} } \right),
\end{equation}
and
\begin{equation}
    E_{ik,s}^{(r)}(t) = {{\frac{{{{p_k}(t)\theta _{ij,s}(t)}}}{{{{\left( {{|g_k(t)-g_{i,s}|^2} + ||{\bm v}_k^{(r)}(t) - {{\bm x}_{i,s}(t)}|{|_2^2}} \right)}^2{{2^{D_{i,s}^{(r)}(t)}\ln 2}}}}}}}.
\end{equation}

Besides, for a given location point $({\bm v}_{j}^{(r)}(t), {\bm v}_{k}^{(r)}(t))$, we can obtain the lower-bound of ${{\left\| {{{\bm v}}_{j}}(t)-{{{\bm v}}_{k}}(t) \right\|}^{2}}$ via the first-order Taylor expansion as described below:
\begin{equation}\label{eq:xjk}
\begin{array}{*{20}{l}}
{|| {{{\bm v}_j}(t) - {{\bm v}_k}(t)} ||^2} \ge  - {|| {{\bm v}_j^{(r)}(t) - {\bm v}_k^{(r)}(t)} ||^2} +
2{( {{\bm v}_j^{(r)}(t) - {\bm v}_k^{(r)}(t)} )^{\rm T}}\left( {{{\bm v}_j}(t) - {{\bm v}_k}(t)} \right).
\end{array}
\end{equation}

Similarly, for a given location point ${\bm v}_k^{r}(t)$, $||{{\bm v}_k}(t) - {{\bm x}_{i,s}(t)}||{^2}$ is lower-bounded by
\begin{align}\label{eq:xk}
||{{\bm v}_k}(t) - {{\bm x}_{i,s}(t)}||{^2} \ge  {|| {{\bm v}_k^{(r)}(t) - {\bm x}_{i,s}(t)} ||^2} + 2{( {{\bm v}_k^{(r)}(t) - {\bm x}_{i,s}(t)} )^{\rm T}} \notag \left( {{{\bm v}_k}(t) - {{\bm x}_{i,s}(t)}} \right).
\end{align}

Besides, it can be proved that all the first-order Taylor expansions mentioned above satisfy the conditions in Assumption \ref{assumption_1}.

For any local point ${\mathcal X}^{(r)}(t) = \{{\bm v}_k^{(r)}(t)\}$, by referring to (\ref{eq:Rij})-(\ref{eq:xk}), (\ref{eq:UAV_location_equal_problem_S}) is approximated as (\ref{eq:UAV_location_equal_problem_approximate}). Till now, all non-convex constraints are transformed into convex ones, and (\ref{eq:UAV_location_equal_problem_S}a) is linear. Therefore, the transformed problem is convex. This completes the proof.

\subsection{Proof of Proposition \ref{lemma:lemma_UAV_power}}
For any given slice request acceptance set ${\mathcal A}(t)$ as well as UAVs' locations ${\mathcal X}(t)$, the UAV transmit power of (\ref{eq:MBB_subproblem}) can be optimized via mitigating the following problem
\begin{subequations}\label{eq:UAV_power_problem}
\begin{alignat}{2}
& \mathop {\rm Maximize }\limits_{{\mathcal {P}}(t),\{ \eta_{i,s}(t) \}} {\mkern 1mu} \text{ } - V\rho \sum\nolimits_{j \in {\mathcal J}} {{p_j}(t)}  - \sum\nolimits_{j \in {\mathcal J}} {{{[{H_j}(t)]}^ + }{p_j}(t)}  + \sum\nolimits_{i,s} {\{ {{[{Q_{i,s}}(t)]}^ + } + {{[{Z_{i,s}}(t)]^+} }\} {\eta _{i,s}}(t)} \\
& {\rm s.t:} \text{ } \sum\limits_{j \in {\mathcal J}} {{a_{ij,s}}(t)W^e(t){{\log }_2}\left( {1 + \frac{{{p_j}(t){h_{ij,s}}(t)}}{{{N_0W^e(t)} + \sum\nolimits_{k \in {\mathcal J}\backslash \{ j\} } {{p_k}(t){h_{ik,s}}(t)} }}} \right)} \ge {\eta _{i,s}}(t), \quad \forall i,s,t \\
& \quad {\rm (\ref{eq:original_problem}d) \text{ } is \text{ } satisfied.}
\end{alignat}
\end{subequations}

Owing to the non-convex constraint (\ref{eq:UAV_power_problem}b), (\ref{eq:UAV_power_problem}) is non-convex; as a result, it is challenging to achieve its optimal solution. However, we observe that (\ref{eq:UAV_power_problem}b) is a difference of two concave functions w.r.t $p_j(t)$, $\forall j$. Accordingly, we adopt the SCA technique again to approximate (\ref{eq:UAV_power_problem}b). Specifically,
$R_{ij,s}(t)$ can be rewritten as ${R_{ij,s}}(t) = {\hat R_{i,s}}(t) - {{\overset{\scriptscriptstyle\smile}{R}}_{ij,s}}(t)$, where ${{\overset{\scriptscriptstyle\smile}{R}}_{ij,s}}(t) = {\log _2}\left( {{N_0W^e(t)} + \sum\nolimits_{k \in {\mathcal J}\backslash \{ j\} } {{p_k}(t){h_{ik,s}}(t)} } \right)$. For any local point ${\mathcal P}^{(r)}(t) = \{p_j^{(r)}(t)\}$, via the first-order Taylor expansion ${{\overset{\scriptscriptstyle\smile}{R}}_{ij,s}}(t)$ is upper-bounded by
\begin{equation}\label{eq:8}
\begin{array}{l}
{{\mathord{\buildrel{\lower3pt\hbox{$\scriptscriptstyle\smile$}}
\over R} }_{ij,s}}(t) \le {\log _2}\left( {{N_0W^e(t)} + \sum\nolimits_{k \in {\mathcal J}\backslash \{ j\} } {p_k^{(r)}(t){h_{ik,s}}(t)} } \right)\\
\qquad \quad + \sum\nolimits_{k \in {\mathcal J}\backslash \{ j\} } {\frac{{{h_{ik,s}}(t)}}{{\left( {{N_0W^e(t)} + \sum\nolimits_{k \in {\mathcal J}\backslash \{ j\} } {p_k^{(r)}(t){h_{ik,s}}(t)} } \right)\ln 2}}} \left( {{p_k}(t) - p_k^{(r)}(t)} \right)\\
\qquad  = F_{ij,s}^{(r)}(t) + \sum\nolimits_{k \in {\mathcal J}\backslash \{ j\} } {G_{ik,s}^{(r)}(t)} \left( {{p_k}(t) - p_k^{(r)}(t)} \right),
\end{array}
\end{equation}
where $F_{ij,s}^{(r)}(t) = {\log _2}\left( {{N_0W^e(t)} + \sum\nolimits_{k \in {\mathcal J}\backslash \{ j\} } {p_k^{(r)}(t){h_{ik,s}}(t)} } \right)$ and $G_{ik,s}^{(r)}(t) = \frac{{{h_{ik,s}}(t)}}{{2^{F_{ij,s}^{(r)}(t)}\ln 2}}$.


We can thus write the lower-bound of $R_{ij,s}(t)$ as ${R_{ij,s}}(t) \ge {\hat R}_{i,s}(t) - F_{ij,s}^{(r)}(t) - \sum\nolimits_{k \in {\mathcal J}\backslash \{ j\} } {G_{ik,s}^{(r)}(t)} ( {{p_k}(t) - p_k^{(r)}(t)} )
$. Besides, it can be proved that the first-order Taylor expansion satisfies the conditions in Assumption \ref{assumption_1}.

In summary, for any local point ${\mathcal P}^{(r)}(t)$, (\ref{eq:UAV_power_problem}) can be approximated as (\ref{eq:UAV_power_problem_approx}). Besides, all non-convex constraints are transformed into convex ones, and (\ref{eq:UAV_power_problem}a) is linear. Therefore, the transformed problem is convex. This completes the proof.

\subsection{Proof of Lemma \ref{lemma:lemma_convergent}}
Denote the opposite of the objective function value of (\ref{eq:MBB_subproblem}) by $\Gamma({{\mathcal A}^{(r)}(t)},{{\mathcal X}^{(r)}(t)},{{\mathcal P}^{(r)}(t)})$ at a local point $({{\mathcal A}^{(r+1)}(t)}, {\mathcal X}^{(r)}(t), {\mathcal P}^{(r)}(t))$.
Then, given a local point $({\mathcal X}^{(r)}(t), {\mathcal P}^{(r)}(t))$, the obtained $\Gamma({{\mathcal A}^{(r+1)}(t)},{{\mathcal X}^{(r)}(t)},{{\mathcal P}^{(r)}(t)})$ via optimizing (\ref{eq:UE_association_problem}) at the $(r+1)$-th iteration is not greater than $\Gamma({{\mathcal A}^{(r)}(t)},{{\mathcal X}^{(r)}(t)},{{\mathcal P}^{(r)}(t)})$. Given a point $({\mathcal A}^{(r+1)}(t), {\mathcal P}^{(r)}(t))$, we have, $\Gamma({{\mathcal A}^{(r+1)}(t)},{{\mathcal X}^{(r)}(t)},{{\mathcal P}^{(r)}(t)}) \ge \Gamma({{\mathcal A}^{(r+1)}(t)},{{\mathcal X}^{(r+1)}(t)},{{\mathcal P}^{(r)}(t)})$ due to the minimization of the upper-bounded problem of (\ref{eq:UAV_location_equal_problem}). Likewise, the inequality $\Gamma({{\mathcal A}^{(r+1)}(t)},{{\mathcal X}^{(r+1)}(t)},{{\mathcal P}^{(r+1)}(t)}) \ge \Gamma({{\mathcal A}^{(r+1)}(t)},{{\mathcal X}^{(r+1)}(t)},{{\mathcal P}^{(r)}(t)})$ can be obtained at $({\mathcal A}^{(r+1)}(t), {\mathcal X}^{(r+1)}(t))$. Besides, $\Gamma({{\mathcal A}^{(r)}(t)},{{\mathcal X}^{(r)}(t)},{{\mathcal P}^{(r)}(t)})$ is bounded at each iteration. Therefore, Algorithm \ref{alg:alg1} is convergent.

Lemma \ref{lemma:1} points out that $\Delta (t)-V( g(t)-\rho \sum\nolimits_{j \in {\mathcal J}}{{{p}_{j}^{\rm tot}}(t)} )$ is upper-bounded at each time slot $t$. The time average of $L(t)$ tends to zero when $t \to \infty$. Therefore, Algorithm \ref{alg:final_algorithm} can make all virtual queues mean-rate stable. This completes the proof.

\ifCLASSOPTIONcaptionsoff
  \newpage
\fi




%
\bibliographystyle{IEEEtran}
\bibliography{Globecom_RAN_slicing}

\end{document}